\documentclass[12pt]{article}

\usepackage[utf8]{inputenc}
\usepackage[T1]{fontenc}
\usepackage{lmodern}
\usepackage{microtype}

\usepackage{geometry}
\usepackage{setspace}
\usepackage{graphicx}
\usepackage{booktabs}
\usepackage{multirow}
\usepackage{float}
\usepackage{algorithm}
\usepackage{algorithmic}

\usepackage{xcolor}
\usepackage{mathtools,amssymb,amsthm,bm}

\theoremstyle{plain}
\newtheorem{theorem}{Theorem}
\newtheorem{lemma}{Lemma}
\newtheorem{proposition}{Proposition}[section]
\newtheorem{corollary}{Corollary}

\theoremstyle{definition}
\newtheorem{definition}{Definition}
\newtheorem{assumption}{Assumption}

\theoremstyle{remark}

\definecolor{ml}{RGB}{34,139,34}

\usepackage[authoryear,round]{natbib}
\usepackage[colorlinks=true, citecolor=blue, linkcolor=blue, urlcolor=blue]{hyperref}
\usepackage{orcidlink}

\title{\textbf{Causal Discovery in Multivariate Extremes via Tail Asymmetry}}

\author{
  Mengran Li\orcidlink{0009-0001-8273-9321}\thanks{%
    School of Mathematics and Statistics, University of Glasgow, UK.
    Correspondence: \texttt{m.li.3@research.gla.ac.uk}}
  \and
  Daniela Castro-Camilo\orcidlink{0000-0002-7536-4613}\thanks{%
    School of Mathematics and Statistics, University of Glasgow, UK.}
}

\date{}

\begin{document}

\maketitle

\begin{abstract}
Causal discovery in multivariate extremes is challenging because extreme observations are sparse,
dependent, and often affected by latent common shocks. Existing approaches focus on undirected extremal dependence, require prior graph restriction, and do not scale beyond small systems.
We introduce \emph{tail-induced asymmetry} as a principle for causal directionality in heavy-tailed systems, where extreme events propagate asymmetrically so that forward tail prediction is systematically easier than backward prediction. We show that this asymmetry yields identifiable causal direction under a canonical max-linear model and provides a basis for score-based structure learning in the tail regime.
Building on this, we propose {Sparse Structure diScovery in Multivariate Extremes} (S3ME), a two-stage data-driven framework for causal discovery.
The first stage performs proxy-adjusted penalized neighbourhood selection to recover a sparse candidate skeleton under latent confounding. The second stage orients edges by minimizing tail prediction risk based on max-linear envelope models, exploiting directional asymmetry.
We establish high-dimensional guarantees for skeleton screening and consistency of the score-based estimator under population separation conditions. Simulations demonstrate robustness to latent confounding and favourable scaling relative to existing extremal methods. Applications to river network data and financial tail-risk networks show that the approach recovers sparse, interpretable propagation structures without prespecified graph structure.

\end{abstract}

\noindent\textit{Keywords:} Extremal dependence graphs, Latent confounding adjustment, Score-based structure learning.

\section{Introduction}

Extreme events often propagate through interconnected systems, making directional inference central
to risk assessment and intervention. In environmental, financial, and infrastructural settings, the
key question is not merely whether variables co-exceed in the tail, but whether an extreme at one
variable helps generate an extreme at another \citep{AsadiDavisonEngelke2015,
GissiblKluppelbergOtto2018}. This question is especially difficult in the tail regime, where
extreme observations are sparse, tail dependence can differ sharply from bulk behaviour, and latent
common drivers can simultaneously push many variables into the tail, inducing strong and largely
symmetric co-exceedance that swamps the directional signatures of local propagation and makes both
skeleton screening and edge orientation unreliable without explicit adjustment \citep{resnick2007heavy}.
Classical causal discovery procedures
such as the PC algorithm \citep{spirtes2000causation}, LiNGAM (Linear Non-Gaussian Acyclic
Model; \citealt{shimizu2006lingam}), and NOTEARS \citep{zheng2018dags} are therefore poorly
matched to this setting. Their inferential logic is calibrated to bulk-level conditional
independence, likelihood structure, or non-Gaussian variation rather than to tail events.

Extreme-value statistics has led to substantial progress in graphical modelling for tail dependence.
Along the \emph{undirected line}, extremal graphical models show that sparse tail dependence can be
represented through graphical structure, making extremal graph learning substantially more tractable
\citep{EngelkeHitz2020}. Subsequent work develops structure learning for extremal tree models and
related extensions of this line \citep{EngelkeVolgushev2022}. Along the \emph{directed line},
approaches have emerged under recursive max-linear or related heavy-tailed structural equation model
formulations. Early work formulates Bayesian-network structure for max-linear systems
\citep{KluppelbergLauritzen2019}. Subsequent contributions study identifiability, estimation, and
causal discovery in recursive max-linear or heavy-tailed models \citep{GissiblKluppelbergLauritzen2021,
GneccoEtAl2021, PascheEtAl2023}.

Taken together, these developments show that directional learning in extremes is possible, but not
yet in the form needed for fully data-driven causal discovery in larger multivariate systems. The
undirected line yields sparse tail structure without causal orientation, whereas the directed line
typically relies on stronger structural restrictions, narrower graph classes, or a prespecified
candidate graph before orientation can begin. In practice, existing directed approaches are largely
confined to small or pre-structured systems. Graph learning typically requires a candidate skeleton
or partial ordering as input, and empirical studies rarely exceed 30 nodes. Learning the graph
entirely from tail data, without prior structural knowledge, in systems with hundreds of variables
and latent confounding, remains out of reach for existing methods.

We address this gap through a structural insight. Extreme-value processes driven by regularly
varying innovations obey the single big jump principle \citep{resnick2007heavy}, under which tail
events are typically generated by one dominant shock rather than by the accumulation of many
moderate contributions. In recursive extremal systems, this creates a directional imbalance.
Conditioning on an extreme at a cause sharply constrains the corresponding effect; conditioning on
an extreme at an effect does not identify its source, since the exceedance may have arisen through
propagation from another variable or through a local innovation. Predicting forward is therefore
systematically easier than predicting backward in the tail regime. We call this phenomenon
\emph{tail-induced asymmetry}. Building on this insight, we propose \emph{Sparse Structure Discovery in Multivariate Extremes}
(S3ME), a two-stage framework for causal discovery in multivariate extremes. The first stage screens for a sparse candidate skeleton via
proxy-adjusted penalized regression, where the proxy adjustment attenuates symmetric dependence
induced by latent common shocks. The second stage orients edges by exploiting tail-induced asymmetry
through a score-based procedure. The framework is supported by an identifiability result for the
canonical bivariate case, high-dimensional consistency guarantees for both stages, and a robustness
analysis showing that proxy adjustment suppresses confounding-induced false positives where existing
directed methods deteriorate. We apply the method to 31 gauging stations locations in the upper Danube 
river network and a tail-risk network of 103 S\&P~500 constituents spanning all eleven GICS sectors, 
settings substantially larger than those in existing directed extremal work. Our method recovers sparse, 
interpretable propagation structures without any prior graph specification.

The paper makes three contributions. First, we identify tail-induced asymmetry as an exploitable
source of causal directionality in heavy-tailed systems, a mechanism that has not previously been
exploited for graph orientation in high-dimensional, fully data-driven settings. Second, we show
that this asymmetry supports consistent edge orientation in high-dimensional settings without a
prespecified graph, extending the reach of causal discovery in extremes well beyond existing
methods. Third, we develop a proxy-adjustment procedure that remains stable under latent confounding
at a scale where competing directed approaches accumulate large numbers of spurious edges.

The remainder of the paper is organized as follows. Section~\ref{sec:identifiability} establishes
the core identifiability result based on tail-induced asymmetry. Section~\ref{sec:methodology}
presents the two-stage methodology. Section~\ref{sec:theory} provides theoretical guarantees for
high-dimensional consistency and robustness to latent confounding. Section~\ref{sec:simulation}
evaluates finite-sample performance through simulations, Section~\ref{sec:applications} applies the
method to hydrological and financial data and Section \ref{sec:discussion} concludes the paper.
Reproducible code is available at \url{https://github.com/MengranLi-git/S3ME}.

\section{Identifiability via Tail Asymmetry}\label{sec:identifiability}
\subsection{Recursive Extremal SEMs and Tail Prediction Risk}

In complex systems, tail events may propagate along directed pathways. To represent this mechanism,
we consider a recursive extremal structural equation model.

\begin{assumption}[Recursive Extremal SEM]
\label{assum:recursive_sem}
Let $X=(X_{1},\dots,X_{p})^{\top}$ be a strictly positive random vector. There exists a directed
acyclic graph (DAG) $\mathcal{G}=(V,E)$ with parent sets $\text{pa}(j)$ such that, for
$j=1,\dots,p$,
\begin{equation}
    X_{j}=f_{j}\!\big(\{X_{i}:i\in \text{pa}(j)\}\big)+\epsilon_{j},
\end{equation}
where: (i) $f_{j}(\cdot)$ is non-decreasing in each argument and homogeneous of order $1$, i.e.,
$f_{j}(cx)=c f_{j}(x)$ for all $c>0$; 
(ii) $\epsilon_{1},\dots,\epsilon_{p}$ are mutually independent, strictly positive, and regularly
varying with common tail index $\alpha>0$.
\end{assumption}

Assumption~\ref{assum:recursive_sem} includes standard recursive extremal SEMs, including max-linear
and extremal additive constructions, and implies multivariate regular variation for $X$. The single
big jump principle then creates a characteristic directional imbalance in the tail, whereby conditioning on
an extreme at a parent can substantially reduce uncertainty about its descendants, whereas
conditioning on an extreme at a child does not uniquely localize its source
\citep[p.~217]{resnick2007heavy}. We formalize this imbalance through tail prediction risk. Let $Z$
denote the tail-domain representation, whose explicit construction is given in
Section~\ref{sec:tail_representation}. We call a function $h:\mathbb{R}\to\mathbb{R}$ a max-linear
envelope predictor if it has the form $h(z)=\max(z,c)+d$ for some $c,d\in\mathbb{R}$, and write
$\mathcal H$ for the class of all such predictors.\label{sec:asymmetry}

\begin{definition}[Tail prediction risk]
\label{def:tailrisk}
At threshold $u$, define the tail prediction risk of $Z_j$ from $Z_i$ by
\begin{equation}
R^{\star}_{j\mid i}(u)
\;:=\;
\inf_{h\in\mathcal H}
\mathbb{E}\!\left[|Z_{j}-h(Z_i)|\,\Big|\,Z_{i}>0\right],
\qquad i\neq j.
\end{equation}
\end{definition}

For nodewise parent-set analysis, we extend Definition~\ref{def:tailrisk} as follows. For any
nonempty candidate parent set $A\subseteq V\setminus\{j\}$, define the multivariate max-linear
envelope class
\[
\mathcal H_A
:=
\Bigl\{h_A(z_A)=\max\!\bigl(c_0,\,\max_{m\in A}(z_m+c_{m\to j})\bigr):\ c_0\in\mathbb R,\ c_{m\to
j}\in\mathbb R\Bigr\},
\]
and the corresponding population tail prediction risk
\[
R^{\star}_{j\mid A}(u)
:=
\inf_{h_A\in\mathcal H_A}
\mathbb{E}\!\left[\,|Z_j-h_A(Z_A)|\,\middle|\,\max_{m\in A} Z_m>0\right].
\]
For $A=\{i\}$, this reduces (up to reparametrization) to the bivariate definition above. We say
that $(i,j)$ exhibits \emph{tail-induced asymmetry} at level $u$ if
$R^{\star}_{j\mid i}(u) \;<\; R^{\star}_{i\mid j}(u).$
Intuitively, when a causal relationship $i \to j$ exists, predicting $j$ from $i$ is
easier than the reverse. A large shock at $i$ reliably drives $j$ into the tail, so $Z_i$ is a 
tight predictor
of $Z_j$. Conversely, backward prediction is more ambiguous. An extreme at $j$ may have been
driven by a shock at $i$ or may have arisen from $j$'s own innovation $\epsilon_j$, and
given only $Z_j$, the two are indistinguishable. The prediction risk therefore satisfies
$R^{\star}_{j\mid i} < R^{\star}_{i\mid j}$, and this asymmetry reveals the causal
direction in a way that standard tail dependence measures, which are symmetric in $i$ and
$j$, cannot.

\subsection{Asymptotic Identifiability}\label{sec:main_identifiability}

While Assumption~\ref{assum:recursive_sem} allows for general regularly varying innovations with
index $\alpha>0$, the directional tail signatures we study are invariant under monotone marginal
transformations. For the theoretical analysis that follows
(Theorem~\ref{thm:asymmetry_identifiability} and Proposition~C.1 in Appendix~\ref{app:proof_prop3}), it is
standard practice to work with the canonical representation. Therefore, without loss of generality,
we assume the innovations have been standardized to unit Fr\'echet scale ($\alpha=1$) via the
probability integral transform, so that $F_{\epsilon_j}(x)=\exp(-1/x)$ for each $j$. This
specification provides a standard heavy-tailed benchmark in which the asymmetry mechanism can be
analyzed in closed form, and it serves as the basic identifiability template for the orientation
step developed later.

\begin{theorem}[Bivariate identifiability via tail asymmetry]
\label{thm:asymmetry_identifiability}
Let $X=(X_{1},X_{2})^{\top}$ follow the bivariate recursive max-linear model
$X_{1}=\epsilon_{1}$ and $X_{2}=\max(cX_{1},\epsilon_{2})$, where $\epsilon_{1},\epsilon_{2}$ are
independent Fr\'echet$(1)$ innovations and $c>0$.
Let $Y_j$ denote the exact unit-Pareto standardization of $X_j$, and let $Z_{j}=\log Y_{j}-\log u$
denote the corresponding log-tail coordinates. Then the forward risk converges to zero, i.e.,
$R^{\star}_{2\mid 1}(u) \;\to\; 0,$  $\text{as } u\to\infty,$
whereas the reverse risk is strictly bounded away from zero:
\begin{equation}
    \liminf_{u\to\infty}\, R^{\star}_{1\mid 2}(u) \;\ge\; \frac{\log 2}{c+1} \;>\; 0.
\end{equation}
Consequently, any monotone $\ell_1$-based scoring criterion evaluated at a sufficiently large
threshold $u$ strictly prefers the true direction $1\to 2$ over $2\to 1$, yielding asymptotic
identifiability under max-linear envelope prediction.
\end{theorem}

The two parts of Theorem~\ref{thm:asymmetry_identifiability} (whose proof is deferred to Appendix~\ref{app:proof_asymmetry_identifiability}) reflect the asymmetry described above. In the forward
direction, when $X_1$ is very large, $X_2 = \max(cX_1, \epsilon_2) \approx cX_1$
dominates, so $Z_2 \approx Z_1 + \log c$ and the prediction error vanishes. In the
reverse direction, an extreme $X_2$ could have arisen from a large $X_1$ or from a large
$\epsilon_2$, and the two scenarios imply very different values of $X_1$; this residual
ambiguity is what keeps $R^{\star}_{1\mid 2}$ bounded away from zero. The lower bound $\log 2/(c+1)$ decreases as $c$ increases, since a stronger causal link
leaves a larger footprint on $X_2$, making the source somewhat easier to recover, but the
ambiguity never fully disappears at any finite $c$.

\subsection{Nodewise and Global Identifiability}
\label{sec:global_identifiability}

Theorem~\ref{thm:asymmetry_identifiability} establishes the directional mechanism in the canonical
bivariate case. We now extend this to the multivariate setting. The key challenge is that each node
may have multiple candidate parents, so identifiability must simultaneously address two questions:
underfitting (omitting a true parent) and overfitting (including spurious variables). These are
controlled by two assumptions.

\begin{assumption}[Non-vanishing conditional prediction error for omitted parents]
\label{assum:cond_asymmetry}
Let $\mathcal G^*=(V,E^*)$ be the true DAG, fix a node $j\in V$, let $\mathrm{pa}^*(j)$ denote its
true parent set, and let $\mathrm{nb}^*(j)=\{i:\{i,j\}\in E^*_{\mathrm{skel}}\}$ denote the
neighbors of $j$ in the true undirected skeleton. For any candidate parent set $A \subseteq
\mathrm{nb}^*(j)$ with $A \not\supseteq \mathrm{pa}^*(j)$, the omission of at least one true parent
induces a non-vanishing increase in population prediction error in the tail. Specifically, the
optimal $\ell_1$ tail prediction risk for node $j$ given parent set $A$ satisfies
$\liminf_{u \to \infty} \bigl[ R^\star_{j \mid A}(u) - R^\star_{j \mid \mathrm{pa}^*(j)}(u) \bigr] >
0.$

This assumption abstracts the multivariate extension of the bivariate tail asymmetry established in
Theorem~\ref{thm:asymmetry_identifiability}: conditioning on a strict subset of the true parents
does not eliminate the directional signal from the missing parent.
\end{assumption}

To operationalize the tail prediction risk for graph learning, we fix a sufficiently large threshold
$u$ where the asymptotic tail asymmetry of Assumption~\ref{assum:cond_asymmetry} takes effect and
work directly with $R^\star_{j\mid A}(u)$ as the population target. Thus, at the population level, we do
not include sample-size-dependent complexity penalties; those enter only in the empirical score used
for finite-sample orientation in Section~\ref{sec:orientation}.

\begin{assumption}[No population risk gain from spurious supersets]
\label{assum:spur_penalty}
Under the same notation, if $A \supsetneq \mathrm{pa}^*(j)$, then adding non-parent variables does
not improve the population tail prediction risk $R^*_{j\mid A}(u)\ \ge\ R^*_{j\mid\mathrm{pa}^*(j)}(u)$
for sufficiently large $u$. In the max-linear benchmark, Proposition~B.1 (Appendix~\ref{app:proof_multivariate_identifiability})
verifies the sharper equality case.
\end{assumption}

Assumption~\ref{assum:cond_asymmetry} controls underfitting, ensuring that omitting a true parent incurs a
strictly positive population risk increase. Assumption~\ref{assum:spur_penalty} controls
overfitting, ensuring that adding spurious parents cannot improve population risk. Together they imply that the
true parent set is the unique inclusion-minimal population minimizer at each node.

\begin{theorem}[Nodewise identifiability]
\label{thm:nodewise_identifiability}
Let $X=(X_1,\dots,X_p)^\top$ follow a recursive max-linear SEM on a true DAG $\mathcal G^*=(V,E^*)$,
and let $E^*_{\mathrm{skel}}$ denote the corresponding undirected skeleton. Fix a node $j\in V$.
Under Assumptions~\ref{assum:cond_asymmetry} and \ref{assum:spur_penalty}, for any fixed and
sufficiently large threshold $u$, the true parent set $\mathrm{pa}^*(j)$ is the unique
inclusion-minimal minimizer of $R^*_{j\mid A}(u)$ over all skeleton-compatible candidate parent sets
$A \subseteq \mathrm{nb}^*(j)$.
\end{theorem}

\begin{corollary}[Global identifiability by risk decomposability]
\label{cor:global_identifiability}
Let $R^*(G; u)=\sum_{j=1}^p R^*_{j\mid \mathrm{pa}_G(j)}(u)$
be the population decomposable risk over DAGs compatible with the true skeleton. If Theorem~\ref{thm:nodewise_identifiability} 
holds for every node $j\in V$, then $\mathcal G^*$ is
the unique inclusion-minimal minimizer of $R^*(G;u)$ over the skeleton-compatible class.
\end{corollary}

The formal proofs are deferred to Appendix~\ref{app:proof_multivariate_identifiability}.
Proposition~B.1 (Appendix~\ref{app:proof_multivariate_identifiability}) verifies the superset non-improvement property in the
max-linear subclass; extending this verification to broader homogeneous non-decreasing SEMs remains
open.

\section{Methodology}\label{sec:methodology}

\subsection{Skeleton Recovery in the Tail Domain}\label{sec:skeleton}

\label{sec:tail_representation}Standard max-linear specifications satisfying
Assumption~\ref{assum:recursive_sem} induce multivariate regular variation on $X$
\citep{KluppelbergLauritzen2019}, placing it in the maximum domain of attraction of a max-stable
distribution. We standardize each margin to unit Pareto scale via the probability integral transform $Y_j = \frac{1}{1-F_j(X_j)},$
where $F_j$ denotes the marginal distribution function of $X_j$. While
Section~\ref{sec:identifiability} works with unit Fr\'echet innovations for theoretical
tractability, the unit Pareto standardization here follows the extremal graphical modeling
convention of \citet{EngelkeHitz2020}; the two scales are tail-equivalent (both regularly
varying with index~1), and the log-transformation in \eqref{eq:tail_vector} yields
approximately exponential tail coordinates that are convenient for subsequent estimation. We then focus on the tail region
defined by the exceedance event $\{\|Y\|_\infty > u\}$, where $\|Y\|_\infty = \max_{1\le
j\le p} Y_j$, for a sufficiently large threshold $u$, ensuring that at least one component
is extreme; this is the standard conditioning event in the
multivariate regular variation framework. Following the extremal graphical modeling framework of
\citet{EngelkeHitz2020}, we introduce the log-transformed tail coordinates
\begin{equation}
    Z = \log Y - \log u,
    \label{eq:tail_vector}
\end{equation}
which provide a convenient representation for characterizing extremal conditional dependence in high
dimensions. Note that $\{\|Y\|_\infty>u\}$ is equivalent to $\{Z\in\mathcal L\}$ with $\mathcal
L=\cup_{k=1}^p\{z:z_k>0\}$.

To make orientation feasible in high dimensions, we first estimate a sparse undirected candidate
skeleton. This stage serves a screening role analogous to sure independence screening in
high-dimensional regression, with the objective of reducing the edge set to a manageable superset that
retains all true edges with high probability, not to recover the skeleton exactly or to estimate
regression coefficients precisely. Removing implausible pairs before orientation keeps the
subsequent direction search tractable, while the sure screening guarantee
(Theorem~\ref{thm:skeleton_consistency}) ensures that the true DAG remains reachable. For this
purpose, we adopt the H\"usler--Reiss model as a working model for the tail law of $Z$, following
\citet{EngelkeHitz2020}.

\begin{assumption}[H\"usler--Reiss tail approximation]
\label{ass:hr}
The law of $Z$ restricted to $\mathcal L$ is approximated by a H\"usler--Reiss extremal model
parameterized by a variogram matrix $\Gamma\in\mathbb{R}^{p\times p}$. Through its associated
precision structure $\Theta$, the model admits a Gaussian graphical representation in the sense that
extremal conditional independence among the components of $Z$ corresponds to zero entries of
$\Theta$ \citep{EngelkeHitz2020}.
\end{assumption}

In practice, the H\"usler--Reiss approximation is applied after partialling out a proxy $P$ for latent common shocks, that is, unobserved variables that simultaneously
drive multiple nodes into the tail, so that edges in the estimated graph reflect direct
tail connections between variables rather than spurious co-exceedance driven by such
shocks. Under
Assumption~\ref{ass:hr}, the conditional distribution of $Z_j$ given $Z_{-j}$ in the tail admits a
regression representation whose coefficients are determined by $\Theta$ \citep{EngelkeHitz2020};
zeros in $\Theta$ thus correspond to zero nodewise regression coefficients, which motivates
$\ell_1$-penalized neighbourhood selection in the spirit of \citet{meinshausen2006high}.

\label{sec:proxy_adjustment}When latent common shocks are present, direct application of nodewise
$\ell_1$-penalized regression can produce dense candidate graphs, because such shocks induce strong
and largely symmetric tail co-exceedance across many pairs of nodes. In such settings, a
low-dimensional proxy $P$ for latent common shocks can be included as an unpenalized control in each
nodewise regression; as shown in Section~\ref{sec:simulation}, this adjustment substantially
improves recovery when confounding is present while leaving performance essentially unchanged when
it is not.
When a proxy is used, let $\{(Z^{(i)},P^{(i)})\}_{i=1}^k$ denote the log-exceedances and
corresponding proxy values. The regression representation of the H\"usler--Reiss model
(Assumption~\ref{ass:hr}) implies that $Z_j$ is approximately linear in $Z_{-j}$ in the
tail, so squared-error loss is a natural fit. Sparsity is enforced by an $\ell_1$ penalty
on the neighbour coefficients $\beta$, following the neighbourhood-selection principle of
\citet{meinshausen2006high}. The proxy $P$ and intercept $\alpha$ enter without a penalty,
so that common-shock variation is always absorbed regardless of the regularisation
strength. For each $j\in V$, we estimate
\begin{equation}
(\hat{\beta}^{(j)},\hat{\gamma}_j,\hat{\alpha}_j)
=\arg\min_{\beta\in\mathbb{R}^{p-1},\,\gamma\in\mathbb{R}^d},\,\alpha\in\mathbb{R}
\left\{
\frac{1}{2k}\sum_{i=1}^{k}\left(Z_{j}^{(i)}-\alpha-\sum_{m\neq j}\beta_{jm}Z_{m}^{(i)}-\gamma^\top P^{(i)}\right)^2
+\lambda_j\|\beta\|_1
\right\},
\label{eq:proxy_nodewise}
\end{equation}
so that the proxy term and intercept are not penalized while sparsity is enforced only on the
cross-node coefficients. The skeleton estimate $\hat{E}_{\mathrm{skel}}$ is constructed by thresholding, retaining
predictor $m$ for node $j$ if $|\hat{\beta}_m^{(j)}|>\tau$ for a threshold $\tau\propto\lambda$,
and taking the union of retained edges across all nodes.

\subsection{Score-Based Edge Orientation}\label{sec:orientation}

Given $\hat{E}_{\mathrm{skel}}$, we orient edges by minimizing a decomposable nodewise score over
DAGs compatible with the candidate graph. The score is designed to exploit the asymmetry identified in
Section~\ref{sec:main_identifiability}, so that directions that better capture extremal
propagation yield smaller tail-domain prediction error.

For a candidate parent set $\mathrm{pa}(j)$ (and proxy $P$ if used), we represent node $j$ by a
max-linear envelope in the log-tail domain. For each candidate parent $k\in \mathrm{pa}(j)$, the transmission offset $c_{k\to j}$ (i.e.,
the log-scale shift in tail level as a shock propagates from $k$ to $j$) is estimated by a low empirical quantile: $c_{k \to j} = \hat{Q}_q(Z_{\cdot j} - Z_{\cdot k})$ for $q\in(0,1)$,
and the baseline intercept $c_0$ is defined analogously. For the proxy, each component $r=1,\ldots,d$ of $P$ has its own offset $c_{r\to j} = \hat{Q}_q(Z_{\cdot j} - P_{\cdot r})$. The fitted
envelope for node $j$ at observation $t$ is then
\begin{equation}
\hat{Z}_{tj} = \max \left\{ c_0,\; \max_{\ell \in \mathrm{pa}(j)} \bigl(Z_{t\ell} + c_{\ell \to j}\bigr),\; \max_{r=1,\ldots,d}\bigl(P_{tr} + c_{r \to j}\bigr) \right\},
\label{eq:maxlinear_envelope}
\end{equation}
where $\max_{\ell\in\emptyset}(\cdot):=-\infty$ by convention, so the parent term is simply absent when $\mathrm{pa}(j)=\emptyset$.
The envelope captures the dominant upstream contribution, with the proxy term absorbing common
forcing.

For each node $j$, we evaluate a candidate parent set by the sum of absolute envelope (SAE) residuals, the
empirical counterpart of the $\ell_1$ tail prediction risk in Definition~\ref{def:tailrisk}:
\begin{equation}
\mathrm{SAE}_j = \sum_{t=1}^k \bigl| Z_{tj} - \hat{Z}_{tj} \bigr|,
\end{equation}
and combine this fit measure with an Extended Bayesian Information Criterion
\citep[EBIC;][]{FoygelDrton2010} complexity penalty. The local score is defined by
\begin{equation}
\mathrm{Score}(j, \mathrm{pa}(j)) = \frac{k}{2} \log\!\left(\frac{\mathrm{SAE}_j}{k}\right) + \frac{1}{2}\bigl(\log k + 2\gamma_{\mathrm{EBIC}}\log p\bigr)\,|\mathrm{pa}(j)|,
\end{equation}
where $p$ is the number of observed variables, $\mathrm{SAE}_j > 0$ by construction (since the
max-linear envelope predictor generically differs from $Z_{tj}$ on a positive-measure set of tail
observations), and the proxy $P$ is excluded from the parent-set penalty. Because the score decomposes across nodes, we optimize the global score $\sum_{j=1}^p
\mathrm{Score}(j, \mathrm{pa}(j))$\label{sec:greedy_dag_search} by greedy search over DAGs
compatible with $\hat{E}_{\mathrm{skel}}$. Starting from the undirected skeleton, the algorithm
iteratively applies the admissible edge addition, deletion, or reversal that yields the largest
score reduction. Moves that create directed cycles or violate the maximum in-degree constraint are
excluded. The search terminates when no admissible move improves the score, and the resulting graph
is returned as the estimated DAG. The consistency theorem below, however, is stated for the oracle global minimizer,
that is, the DAG that exactly minimizes the empirical score over the admissible class
regardless of computational cost; the greedy search should therefore be viewed as a
computational approximation to that target.

\subsection{The Proposed Algorithm}

Algorithm~\ref{alg:main} summarises the full S3ME procedure. When no proxy is available, the proxy offset terms $c_{r\to j}$ and components $P^{(i)}_r$ ($r=1,\ldots,d$) are omitted from the offset estimation and envelope computation steps, and the max-linear envelope simplifies to
$\hat{Z}^{(i)}_j = \max\Bigl\{c_0,\ \max_{\ell\in\mathrm{pa}(j)}\bigl(Z^{(i)}_\ell + c_{\ell\to j}\bigr)\Bigr\}.$
The orientation step implements a greedy approximation to the oracle global score minimiser in Theorem~\ref{thm:orientation_consistency}. While the consistency guarantee in that theorem is established for the exact minimiser, the greedy search provides a computationally feasible surrogate in practice. The skeleton stage requires solving \(p\) independent nodewise Lasso problems, each with computational cost \(O(kp)\), yielding an overall complexity of \(O(kp^2)\). In the orientation stage, the greedy search examines at most \(O(|\hat{E}_{\mathrm{skel}}| + p)\) candidate moves per iteration. Offset estimation and SAE evaluation cost \(O(kd_{\max})\) per affected node, and offsets corresponding to unchanged parent-child pairs can be cached across iterations. The bound \(d_{\max}\) therefore plays a dual role, controlling the per-move computational cost while also acting as a structural regulariser on the estimated DAG. In practice, the sparsity guarantee \( |\hat{E}_{\mathrm{skel}}| = O(ps_{\max}) \) from Theorem~\ref{thm:skeleton_consistency} helps keep the orientation search tractable in moderate dimensions. Two implementation details are worth noting. First, the marginal standardization $\hat{Y}_j^{(i)} = 1/(1-\hat{F}_j(X_j^{(i)}))$ uses the rescaled empirical CDF $\hat{F}_j(x) = \frac{1}{n+1}\sum_{i'=1}^n \mathbf{1}(X_j^{(i')}\le x)$ (pseudo-observations with denominator $n+1$), so that $\hat{Y}_j^{(i)} = (n+1)/(n+1-R_j^{(i)})$ where $R_j^{(i)}$ is the rank of $X_j^{(i)}$; this ensures $\hat{Y}_j^{(i)}<\infty$ for all observations. Second, the orientation search is initialised from a DAG obtained by assigning an arbitrary acyclic orientation to the skeleton edges, so that all subsequent moves operate on a directed graph.
% Theorem~\ref{thm:skeleton_consistency} helps keep the orientation search tractable in moderate dimensions.

\begin{algorithm}[H]
\caption{S3ME: Sparse Structural Causal Discovery in Multivariate Extremes}
\label{alg:main}
\renewcommand{\algorithmicensure}{\textbf{Output:}}
\begin{algorithmic}[1]
\REQUIRE Observations $\{X^{(i)}\}_{i=1}^n$ with $X^{(i)}\in\mathbb{R}^p$, proxy values $\{P^{(i)}\}_{i=1}^n$ with $P^{(i)}\in\mathbb{R}^d$, tail threshold $u$, regularisation parameters $\{\lambda_j\}$, coefficient threshold $\tau$, offset quantile $q \in (0,1)$, EBIC parameter $\gamma_{\mathrm{EBIC}}$, maximum in-degree $d_{\max}$.
\ENSURE Estimated DAG $\hat{G}$.
\STATE \textbf{Tail Preprocessing.}
\STATE Compute $\hat{Y}_j^{(i)} = 1/(1-\hat{F}_j(X_j^{(i)}))$ via the empirical rank transform.
\STATE Retain the $k$ observations with $\|\hat{Y}^{(i)}\|_\infty > u$, re-index them as $i=1,\ldots,k$, and set $Z^{(i)} = \log\hat{Y}^{(i)} - \log u$; retain the corresponding $\{P^{(i)}\}_{i=1}^k$.
\STATE \textbf{Skeleton Estimation.}
\FOR{$j = 1, \ldots, p$}
\STATE Solve the proxy-adjusted nodewise Lasso with penalty $\lambda_j$ on $\beta$ and unpenalised proxy $P$ to obtain $\hat{\beta}^{(j)}$.
\ENDFOR
\STATE Form $\hat{E}_{\mathrm{skel}}$ with edges $\{j,\ell\}$ if $\max(|\hat{\beta}_\ell^{(j)}|,\,|\hat{\beta}_j^{(\ell)}|)>\tau$.
\STATE \textbf{Score-Based Orientation.}
\STATE Initialise $\hat{G}$ as the empty DAG on vertex set $V$.
\REPEAT
\FOR{each admissible move $a$ (edge addition, deletion, or reversal in $\hat{G}$, with additions restricted to pairs in $\hat{E}_{\mathrm{skel}}$, preserving acyclicity and in-degree $\le d_{\max}$)}
    \FOR{each node $j$ whose parent set changes under $a$}
        \STATE Estimate offsets $c_{\ell\to j} = \hat{Q}_q\bigl(\{Z^{(i)}_j - Z^{(i)}_\ell\}_{i=1}^{k}\bigr)$ for each $\ell\in\mathrm{pa}(j)$; set $c_{r\to j} = \hat{Q}_q\bigl(\{Z^{(i)}_j - P^{(i)}_r\}_{i=1}^k\bigr)$ for each $r=1,\ldots,d$; and $c_0 = \hat{Q}_q\bigl(\{Z^{(i)}_j\}_{i=1}^k\bigr)$.
        \STATE Compute $\hat{Z}^{(i)}_j = \max\bigl\{c_0,\;\max_{\ell\in\mathrm{pa}(j)}(Z^{(i)}_\ell+c_{\ell\to j}),\;\max_{r=1,\ldots,d}(P^{(i)}_r+c_{r\to j})\bigr\}$ for each $i=1,\ldots,k$.
        \STATE Evaluate $\mathrm{SAE}_j = \sum_{i=1}^{k} |Z^{(i)}_j - \hat{Z}^{(i)}_j|$ and $\mathrm{Score}_{\gamma_{\mathrm{EBIC}}}(j,\mathrm{pa}(j))$.
    \ENDFOR
\ENDFOR
\STATE Set $a^{*} = \arg\min_{a}\,\Delta\mathrm{Score}_{\gamma_{\mathrm{EBIC}}}(a)$.
\IF{$\Delta\mathrm{Score}_{\gamma_{\mathrm{EBIC}}}(a^{*}) < 0$}
    \STATE Apply $a^{*}$ to $\hat{G}$.
\ENDIF
\UNTIL{no admissible move reduces the score.}
\STATE \textbf{return} $\hat{G}$.
\end{algorithmic}
\end{algorithm}

\section{Theoretical Properties}\label{sec:theory}

Let $\mathcal{G}^{*}=(V,E^{*})$ be the true causal DAG with node set $V=\{1,\dots,p\}$, and let
$E^{*}_{\mathrm{skel}}=\bigl\{\{i,j\}: (i,j)\in E^{*}\ \text{or}\ (j,i)\in E^{*}\bigr\}$
denote the corresponding undirected skeleton. The theoretical analysis reflects the division of
labour between the two stages, where the skeleton procedure is a screening device that retains all true
edges with high probability, while the orientation step eliminates spurious edges and recovers the
correct directions through the EBIC penalty and tail-induced asymmetry respectively.

\subsection{Sure Screening for the Skeleton}
\label{sec:thm_skel}

We work in a high-dimensional extreme-value regime in which both the dimension $p$ and the
effective tail sample size $k$ (the number of exceedances retained for inference) tend to infinity,
potentially allowing $p\gg k$. We assume $k=k(n)$ is an intermediate sequence with $k\to\infty$ and
$k/n\to 0$ as $n\to\infty$.

Under the H\"usler--Reiss working model (Assumption~\ref{ass:hr}), the precision structure $\Theta$
introduced in Section~\ref{sec:skeleton} induces a covariance-type matrix $\Sigma=\Theta^{-1}$. As
shown by \citet{EngelkeHitz2020}, the nodewise regression coefficients in the tail can be expressed
analogously to the Gaussian case as $\beta_{jk}=-\Theta_{jk}/\Theta_{jj}$, so zeros in $\beta$
encode zeros in $\Theta$.

\begin{assumption}[Regularity conditions for sure screening]
\label{assum:regularity_conditions}
Under the H\"usler--Reiss working model and the proxy-adjusted nodewise regression, assume:
(i) $s_{\max}$ is sufficiently small relative to $k/\log p$,
(ii) the relevant submatrices of $\Sigma=\Theta^{-1}$ are uniformly well-conditioned,
(iii) the minimum non-zero tail regression coefficient is above the screening noise level, and
(iv) the nodewise score bias from the log-sum-exp approximation is of order at most $\lambda$.
Precise constants and rate thresholds are deferred to Appendix~\ref{app:proof_prop3}.
\end{assumption}

Assumption~\ref{assum:regularity_conditions} provides the extremal analogs of the standard
high-dimensional conditions used in neighborhood selection \citep{meinshausen2006high,
wainwright2009sharp}. We do not require the irrepresentability condition needed for exact support recovery
\citep[Condition~1]{wainwright2009sharp}; rather, the conditions above ensure stable local design,
detectable tail signals, and controlled approximation bias at effective sample size $k$.

In the tail regime, the log-sum-exp proxy is a uniformly bounded approximation of the max-linear
structural equation; a formal statement and proof are given in
Proposition~C.1 (Appendix~\ref{app:proof_prop3}) in Appendix~\ref{app:proof_prop3}.

\begin{theorem}[Sure screening for the skeleton]
\label{thm:skeleton_consistency}
Suppose Assumptions~\ref{ass:hr} and~\ref{assum:regularity_conditions} hold,
and the tail linearization error is uniformly bounded in the sense of
Proposition~C.1 (Appendix~\ref{app:proof_prop3}). Let $\hat{E}(\lambda)$ denote the undirected edge set
obtained by the union of selected neighborhoods from nodewise $\ell_1$-penalized regression
(Section~\ref{sec:skeleton}), with the proxy included as an unpenalized regressor if used, using
regularization parameter $\lambda>C\sqrt{\log p/k}$ for a sufficiently large constant $C>0$, as is
standard in high-dimensional regression theory. Then there exist constants $c_{1},c_{2},c_{3}>0$
such that for sufficiently large $k$ and $p$,
\begin{equation}
    \mathbb{P}\!\left(E^{\star}_{\mathrm{skel}}\subseteq\hat{E}(\lambda)\right)
    \ \ge\ 1-c_{1}\exp(-c_{2}k\lambda^{2}),
\end{equation}
and, on this event, the number of selected edges satisfies $|\hat{E}(\lambda)|\le
c_{3}\,p\,s_{\max}$.
\end{theorem}

Sure screening, rather than exact support recovery, suffices for the subsequent
orientation step. Indeed, the sure screening guarantee ensures that the true DAG $G^*$ belongs
to the search space $\mathcal{G}(\hat{E}_{\mathrm{skel}})$, so the correct orientation
is always reachable, while the bound $|\hat{E}|\le c_3\,p\,s_{\max}$ keeps the search
space manageable. Within this space, the orientation conditions, including a positive
underfitting risk gap $\Delta_{\mathrm{miss}}>0$, superset non-improvement, and
two-sided bounds on admissible population risks (Assumption~\ref{assum:overfit_gain}),
ensure that orienting an edge incorrectly raises the empirical score. Underfitted parent
sets are detectable because the missing-parent risk gap is amplified in the tail, while
spurious parents are excluded by the EBIC penalty. Together, these properties allow the
orientation step to simultaneously discard false positives from the skeleton and recover
the correct edge directions.

\subsection{Consistency of Score-Based Orientation}\label{sec:thm_orientation}

\begin{lemma}[Nodewise uniform empirical risk convergence]
\label{lem:uniform_score}
Assume that, uniformly over all admissible pairs $(j,A)$, the empirical tail prediction risk
$\hat{R}_{j\mid A}(u)=\mathrm{SAE}_j(A)/k$ satisfies
\[
\mathbb{P}\bigl(\bigl|\hat{R}_{j\mid A}(u) - R^*_{j\mid A}(u)\bigr| > \epsilon\bigr) \le c_1
\exp\bigl\{-c_2\, k\min(\epsilon^2,\,\epsilon)\bigr\},
\]
for some constants $c_1, c_2 > 0$ and all $\epsilon > 0$. Under the conditions of Theorem~\ref{thm:skeleton_consistency}
and $\log p = o(k)$,
\[
\max_{1\le j\le p}\sup_{\substack{A \subseteq \mathrm{nb}(j;\hat{E}_{\mathrm{skel}})\\|A|\le
s_{\max}}} \bigl|\hat{R}_{j\mid A}(u) - R^*_{j\mid A}(u)\bigr| =
O_p\!\left(\sqrt{\frac{(s_{\max}+1)\log p}{k}}\right).
\]
\end{lemma}

The proof is deferred to Appendix~\ref{app:proof_uniform_score}.

\begin{assumption}[Equal-risk superset fit-gain control]
\label{assum:overfit_gain}
Let $\hat f_j(A)=\frac{k}{2}\log\hat R_{j\mid A}(u)$ with $\hat R_{j\mid A}(u)=\mathrm{SAE}_j(A)/k$.
We assume the following uniform nested-model regularity condition: for all admissible pairs
$A^*\subsetneq A$ in $\mathcal A_j^{c_3}$ satisfying
$R^*_{j\mid A}(u)=R^*_{j\mid A^*}(u),$
the quasi-likelihood-ratio gain from adding $B=A\setminus A^*$ obeys
\[
\max_{1\le j\le p}\ \sup_{\substack{A^*\subsetneq A,\ A,A^*\in\mathcal A_j^{c_3}\\R^*_{j\mid
A}(u)=R^*_{j\mid A^*}(u)}}
\frac{\hat f_j(A^*)-\hat f_j(A)}{|B|}
=O_p(1).
\]
Equivalently,
$\hat f_j(A^*)-\hat f_j(A)=O_p(|B|)$
uniformly over equal-risk supersets.

\end{assumption}
In practice, the above means that adding $|B|$ spurious parents to the true parent set $A^*$
can improve the empirical fit by at most $O_p(|B|)$, which the EBIC penalty, growing as
$|B|\log p$, dominates in high dimensions. Intuitively, the fit gain from spurious
parents is too small to offset the dimension-dependent penalty.

\begin{theorem}[Consistency of the oracle global empirical score minimizer]
\label{thm:orientation_consistency}
Assume the conditions of Theorem~\ref{thm:skeleton_consistency}. For each node $j$, let
$\mathcal{A}_j^{c_3}$ denote the class of admissible parent sets $A\subseteq V\setminus\{j\}$ with
$|A|\le s_{\max}$ such that $A\subseteq \mathrm{nb}(j;E)$ for some $E\supseteq E^*_{\mathrm{skel}}$
satisfying $|E|\le c_3ps_{\max}$. Suppose the following conditions hold at the chosen threshold
$u$: (i) nodewise uniform underfitting gap,
\[
\inf_{1\le j\le p}\ \inf_{\substack{A\in\mathcal{A}_j^{c_3}\\A\not\supseteq \mathrm{pa}^*(j)}}
\Bigl\{R^*_{j\mid A}(u)-R^*_{j\mid\mathrm{pa}^*(j)}(u)\Bigr\}
\ge \Delta_{\mathrm{miss}} > 0;
\]
(ii) superset non-improvement,
$R^*_{j\mid A}=R^*_{j\mid\mathrm{pa}^*(j)}(u)$ for all $A\supseteq \mathrm{pa}^*(j)$ in
$\mathcal{A}_j^{c_3}$ (verified for the recursive max-linear benchmark by Proposition~B.1 in
Appendix~\ref{app:proof_multivariate_identifiability}); and (iii) two-sided admissible-risk bounds,
\[
0<r_{\min}(u)\le \inf_{1\le j\le p}\inf_{A\in\mathcal A_j^{c_3}}R^*_{j\mid A}(u)
\le \sup_{1\le j\le p}\sup_{A\in\mathcal A_j^{c_3}}R^*_{j\mid A}(u)\le r_{\max}(u)<\infty.
\]
Assume additionally Assumption~\ref{assum:overfit_gain} and
\[
s_{\max}\sqrt{\frac{\log p}{k}}=o(1).
\]
Let
\[
\hat{G} \in \arg\min_{G\in\mathcal{G}(\hat{E}_{\mathrm{skel}})}\hat{S}(G)
\]
be any global empirical EBIC score minimizer over DAGs compatible with $\hat{E}_{\mathrm{skel}}$.
Then
\[
\mathbb{P}(\hat{G}=G^*) \to 1 \qquad \text{as } k,p\to\infty.
\]
\end{theorem}
The proof is deferred to Appendix~\ref{app:proof_orientation_consistency}.
\subsection{Robustness to Latent Confounding via Proxy Adjustment}
\label{sec:thm_proxy}

We now consider the effect of residual confounding on the score-based orientation step. The skeleton
stage already incorporates the proxy as an unpenalized regressor
(Section~\ref{sec:proxy_adjustment}), and Theorem~\ref{thm:skeleton_consistency} guarantees that
true edges are retained under this adjustment. We therefore focus on the residual effect of
imperfect proxy partialling-out on the orientation step, where the concern is that remaining
confounding-induced dependencies could distort the population score ordering.

\begin{assumption}[Vanishing residual confounding]
\label{assum:proxy_bias}
Let $P$ be the proxy introduced in Section~\ref{sec:proxy_adjustment}. For each node $j$, let
$\tilde Z_j = Z_j - \Pi_P(Z_j)$ denote the tail-domain residual after projecting onto the span of
$P$
(e.g., $\Pi_P(Z_j)=\arg\min_{a}\sum_{t}(Z_{t,j}-aP_t)^2$ on the tail sample).
Let $\beta_{ji}^{\star}$ denote the corresponding population regression coefficient of $\tilde Z_j$
on $\tilde Z_i$
under the tail-limit working model.

Let $\mathcal{E}_{\mathrm{conf}}$ denote pairs $(i,j)$ that are $d$-separated in the true DAG
$\mathcal{G}^\star$ but share the latent common driver $U$ (i.e., purely confounded pairs). We
assume that the residual confounding leakage satisfies
\[
|\beta_{ji}^{\star}|\ \le\ \gamma_{\mathrm{conf}}(u) \quad \text{for all } (i,j) \in
\mathcal{E}_{\mathrm{conf}},
\]
where $\gamma_{\mathrm{conf}}(u) \to 0$ as $u \to \infty$. For high-dimensional consistency, we
further assume that $\gamma_{\mathrm{conf}}(u)$ vanishes at a rate faster than the score-separation
scale, for example $\gamma_{\mathrm{conf}}(u) = o(k^{-1/2})$ or more generally
$o(\Delta_{\mathrm{miss}})$ when the nodewise population risks satisfy the uniform underfitting gap
$\Delta_{\mathrm{miss}}>0$. This ensures that confounding-induced spurious correlations are
asymptotically negligible compared to genuine causal signals.
\end{assumption}

Assumption~\ref{assum:proxy_bias} is intended to capture a low-dimensional common-shock mechanism
rather than an arbitrary proxy choice. It is plausible when the latent driver enters many nodes
through a shared monotone factor and the proxy is itself a tail-informative surrogate for that
factor, so that partialling out $P$ removes most of the common variation while leaving local
propagation structure in the residuals \citep{PascheEtAl2023}. The assumption does not require the
proxy to identify the latent driver exactly; rather, it requires the proxy to absorb enough of the
common tail component that residual pairwise dependence among purely confounded nodes becomes
asymptotically negligible relative to the orientation signal. Concrete proxy constructions
satisfying this requirement are discussed in Section~\ref{sec:applications}. Under
Assumption~\ref{assum:proxy_bias}, residual confounding does not overturn the population score
ordering: the conditions of Theorem~\ref{thm:orientation_consistency} remain satisfied, and the
consistency guarantee continues to hold.

\section{Simulation Study}
\label{sec:simulation}

\subsection{Simulation Setup}

We assess the finite-sample performance of the proposed framework under recursive max-linear
structural equation models with heavy-tailed noise and latent confounding. The simulation study is
organized around two questions, namely how performance changes with increasing dimension and how sensitive
the procedure is to increasing confounding intensity. In each setting, results are summarized over
50 Monte Carlo replicates and evaluated through skeleton recovery, edge orientation, and overall DAG
recovery.

For each replicate, we generate a directed acyclic graph $G=(V,E)$ from a Barab\'asi--Albert construction \citep{BarabasiAlbert1999}, with the attachment parameter and confounding design varied across experiments. Given
$G$, the observed variables are generated from the recursive max-linear model
\begin{equation}
X_j = \max \left( \bigvee_{i \in \mathrm{pa}(j)} \{c_{ij} X_i\}, \bigvee_{l \in \mathrm{conf}(j)} \{b_{lj} U_l\}, \epsilon_j \right), \qquad j=1, \dots, p,
\end{equation}
where $\mathrm{pa}(j)$ is the parent set of node $j$, $\mathrm{conf}(j)$ indexes latent confounders
acting on node $j$, and the innovations $\epsilon_j$ and latent drivers $U_l$ are independent
standard Fr\'echet variables. The edge coefficients $c_{ij}$ are sampled from $[0.6,0.9]$, and the
confounding coefficients $b_{lj}$ from $[0.5,0.8]$. For each latent driver $U_l$, we additionally
generate a proxy
$P_l = \max(a_l U_l, \eta_l),$
with $a_l \sim U[0.6,0.9]$ and independent Fr\'echet noise $\eta_l$. Here $q$ denotes the number of
latent confounders, and $\mathrm{conf}_s$ denotes the number of observed nodes affected by each
confounder.

After simulating $n$ observations, each margin is transformed to unit Pareto scale by the empirical
rank transform. We then follow the multivariate GPD exceedance rule and retain observations
satisfying $\max_{1\le j\le p} Y_{tj}>u$, where $u$ is defined by the common quantile level
$q_{\mathrm{conf}}=1-n^{0.7}/n.$
Equivalently, we select rows in which at least one margin exceeds its marginal $q_{\mathrm{conf}}$
threshold. Let $k_{\mathrm{exc}}$ denote the realized number of selected exceedances under this
rule; $k_{\mathrm{exc}}$ is data-adaptive and not fixed a priori. We then work with the log-tail
coordinates $Z = \log Y - \log u$ as in Section~\ref{sec:tail_representation}. Depending on the
experiment, we report precision, recall, $F_1$, structural Hamming distance (SHD), and confounding
false positives. Competing methods and experiment-specific parameter settings are introduced in the
corresponding subsections.

\subsection{High-Dimensional Scaling}

We study scaling with fixed sample size $n=1000$ and dimension $p\in\{20,50,100,200\}$. Graphs are
generated from Barab\'asi--Albert models with $m\in\{1,2\}$ and the default generator setting
$q=\lfloor 0.1p\rfloor$, and confounding exposure width is set to $\mathrm{conf}_s=0.2p$ (that is,
$\{4,10,20,40\}$ across dimensions). The tail sample is selected by the exceedance rule above, so the realized
$k_{\mathrm{exc}}$ is data-adaptive. We use $\lambda=C\sqrt{\log p/k_{\mathrm{exc}}}$ with
pre-specified $C=(0.5,1.0,1.5,2.0)$ for $p=(20,50,100,200)$, respectively, and
$\gamma_{\mathrm{EBIC}}=(10,10,10,20)$, with both choices fixed a priori and a larger penalty at
$p=200$ to enforce stronger complexity control in the largest search space.

Four methods are compared. The S3ME skeleton is the proxy-adjusted nodewise $\ell_1$ estimator of
Section~\ref{sec:skeleton}, evaluated on its own as a screening device. S3ME denotes the full
two-stage procedure comprising greedy EBIC-score orientation and pruning on the S3ME skeleton. EASE
\citep{GneccoEtAl2021} is a tail-conditional independence test applied to orient the same recovered
skeleton, so that the comparison isolates the contribution of the orientation score; it uses
$k=\lfloor n^{0.7}\rfloor$ tail observations and the same marginal transformation, with no pruning
step. FGES \citep{Chickering2002} provides a standard body-data baseline using the Gaussian BIC
score on the full sample of $n=1000$ observations; it is run independently with penalty discount
$1.0$ and faithfulness assumed. Figure~\ref{fig:exp1_boxplot} summarizes $F_1$ and SHD, while
Precision and Recall are reported in the Appendix.

\begin{figure}[h]
\centering
\includegraphics[width=0.49\textwidth,trim=0 0 432 360,clip]{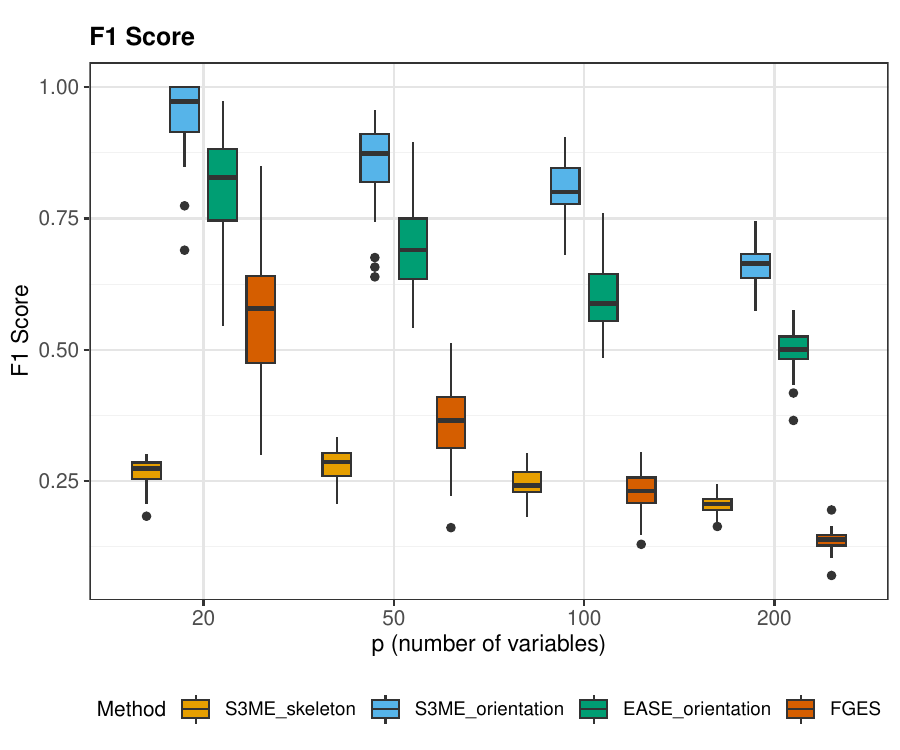}\hfill
\includegraphics[width=0.49\textwidth,trim=432 0 0 360,clip]{Figures/exp1_metrics_boxplot.pdf}
\caption{High-dimensional scaling results ($m=1$, $n=1000$, 50 replicates). Boxplots summarize
$F_1$ (left) and SHD (right) across $p\in\{20,50,100,200\}$ for S3ME skeleton, EASE, FGES, and
S3ME.}
\label{fig:exp1_boxplot}
\end{figure}

All methods degrade as $p$ grows, but the rate of degradation differs markedly. FGES shows a
pronounced deterioration in overall recovery quality, with lower $F_1$ and larger SHD as dimension
increases, because its Gaussian BIC score cannot distinguish tail propagation from body-level
correlation in heavy-tailed data. The S3ME skeleton remains a liberal screening device, so its
$F_1$ stays modest (mean skeleton $F_1$ of $0.27,0.28,0.25,0.21$ for $p=20,50,100,200$) and its
SHD increases with dimension as confounding-induced false positives accumulate. At $p=200$, the
proposed orientation step still outperforms EASE, with mean DAG $F_1$ of $0.66$ for S3ME versus
$0.50$ for EASE, and a correspondingly wider SHD gap in favour of S3ME.

Confounding increasingly dominates skeleton errors in high dimensions, with the confounding share
among skeleton false positives, $\mathrm{ConfFP}_{\mathrm{frac}}$, increasing from $0.08$ at
$p=20$ to $0.92$ at $p=200$. Because EASE orients but does not prune the shared skeleton, its
DAG-stage ConfFP tracks the skeleton-level values exactly. Across $p=20,50,100,200$, DAG-stage ConfFP is
$(0.0,1.2,6.1,16.1)$ for S3ME, compared with $(7.8,92.0,312.5,707.2)$ for EASE and
$(4.1,47.9,245.8,1086.9)$ for FGES. At $p=200$, this corresponds to roughly $44\times$ and
$67\times$ reductions relative to EASE and FGES, respectively. FGES accumulates even more
confounding FPs because it lacks both a tail selection mechanism and a confounding-aware screening
step. This quantitative gap is consistent with the mechanism predicted by
Theorem~\ref{thm:orientation_consistency}, whereby the EBIC penalty discards spurious parents whose
empirical score improvement does not justify the complexity cost, while the tail asymmetry score
ensures that genuinely causal directions are retained.

\subsection{Sensitivity to Confounding Intensity}

This experiment is a controlled stress test that uses the latent-driver signal as proxy input
(oracle-style alignment). We fix $p=100$, set $m=1$, use one latent confounder ($q=1$), and vary
confounding exposure width via $\mathrm{conf}_s\in\{0,0.2p,0.5p,p\}$. All other settings are fixed.

Table~\ref{tab:sim2_combined} compares three methods at the DAG level, namely the full proposed procedure
(S3ME), an ablation that removes the proxy adjustment from both stages (No proxy), and EASE applied
to the S3ME skeleton. For symmetry clarification, S3ME and EASE share the proxy-adjusted skeleton,
while No proxy is an end-to-end ablation using its own no-proxy skeleton. The table reports
directed-stage $F_1$ and confounding-induced false positives (ConfFP).

\begin{spacing}{1}
\begin{table}[h]
\centering
\caption{Confounding sensitivity ($p=100$, $m=1$, 50 replicates). ConfFP denotes
confounding-induced false positives.}
\begin{tabular}{lcccccc}
\toprule
& \multicolumn{2}{c}{S3ME} & \multicolumn{2}{c}{No proxy} & \multicolumn{2}{c}{EASE} \\
\cmidrule(lr){2-3} \cmidrule(lr){4-5} \cmidrule(lr){6-7}
$\mathrm{conf}_s$ & $F_1$ & ConfFP & $F_1$ & ConfFP & $F_1$ & ConfFP \\
\midrule
0      & 0.831 & 0.0 & 0.847 & 0.0  & 0.671 & 0.0   \\
0.2p   & 0.836 & 0.1 & 0.818 & 6.4  & 0.467 & 74.7  \\
0.5p   & 0.837 & 0.5 & 0.767 & 18.9 & 0.294 & 206.9 \\
p      & 0.837 & 1.3 & 0.701 & 39.1 & 0.178 & 402.2 \\
\bottomrule
\end{tabular}
\label{tab:sim2_combined}
\end{table}
\end{spacing}

The pattern is sharp. Without confounding ($\mathrm{conf}_s=0$), the no-proxy ablation is slightly
better ($F_1=0.847$ vs $0.831$). The crossover appears immediately at $\mathrm{conf}_s=0.2p$, where
S3ME overtakes no-proxy ($0.836$ vs $0.818$), and the advantage then widens monotonically with
confounding strength. At $\mathrm{conf}_s=p$, S3ME attains $F_1=0.837$ versus $0.701$ for no-proxy
and $0.178$ for EASE. The same gradient appears in ConfFP, with values $1.3$ (S3ME), $39.1$ (no-proxy), and
$402.2$ (EASE) at $\mathrm{conf}_s=p$. This confirms that the proxy term is not uniformly
beneficial, but is critical for robustness once confounding is present. The ablation further
indicates that the main gain is a cleaner candidate structure, as S3ME keeps far
fewer confounding-induced edges in the final DAG at every confounded setting ($0.1$ vs $6.4$, $0.5$
vs $18.9$, and $1.3$ vs $39.1$), which aligns with its higher $F_1$. We emphasize that this phase uses latent-driver proxy input as a controlled stress test, so the
result should be read as robustness evidence under aligned proxy information rather than as a full
observational-proxy guarantee.

Additional robustness checks are reported in Appendices~\ref{app:threshold_sensitivity} and~\ref{app:graph_robustness}. Appendix~\ref{app:graph_robustness} compares performance
across Erd\H{o}s--R\'enyi and Barab\'asi--Albert graph topologies at matched edge density, showing
that the main conclusions hold for ER graphs while the lower recall observed under BA with $m=2$ is
attributable to hub-induced neighbourhood complexity rather than a failure of the method.
Appendix~\ref{app:threshold_sensitivity} examines sensitivity to the exceedance threshold $\beta$,
confirming that DAG-level performance is stable across $\beta\in\{0.5,\ldots,0.9\}$ and that the
default $\beta=0.7$ lies in a plateau of near-optimal performance.

\section{Applications}\label{sec:applications}

We illustrate the method on two real-data settings with different validation regimes, namely a river
network with a physically interpretable flow structure and a financial network without known ground
truth.

\subsection{Danube River Network}
\label{sec:app1}

We first apply the proposed framework to extreme river discharge data from the upper Danube basin.
We analyze daily summer (June--August) discharge measurements from $n=428$ observations at 31
gauging stations provided by the Bavarian Environmental Agency, restricted to the common observation
period 1960--2010 \citep{AsadiDavisonEngelke2015}. After standard temporal declustering and tail
preprocessing (Appendix~\ref{app:application_preprocessing}), the threshold $q_{\mathrm{conf}}=0.90$ yields $k=117$ tail
observations. This setting is useful because the river
topology, derived from the known catchment structure of the upper Danube
\citep{AsadiDavisonEngelke2015, EngelkeHitz2020}, provides a physically interpretable benchmark for
the direction of extremal propagation.

A central difficulty in this application is the presence of regional meteorological forcing:
large-scale storm systems can drive extreme discharges at multiple stations simultaneously, inducing
strong symmetric tail dependence that obscures the directional propagation along the river network.
To attenuate this common component, we construct a proxy series
$P_t = \sum_{j=1}^{p} \mathbb{I}(Z_{t,j} > 0)$, which records the spatial extent of simultaneous
exceedances across the network and serves as a scalar summary of region-wide meteorological
activity. We then apply the two-stage procedure to the preprocessed data, using
$\lambda = \sqrt{\log(p)/k}$ for skeleton estimation and $\gamma_{\mathrm{EBIC}}=10$ in the
orientation step.

Figure~\ref{fig:danube_network} compares the estimated extremal DAG with the known physical
river-flow graph. Against the 30 true edges in the upstream-to-downstream topology, the method
recovers 24 skeleton edges with precision $0.88$ and recall $0.70$ ($F_1=0.78$). Among the 21
correctly identified edges, $17$ ($0.81$) receive the correct direction and $4$ ($0.19$) are
reversed. The reversed edges occur at tributary confluence points, where the max-linear signal from
the main stem can dominate the tributary contribution and make the direction ambiguous from tail
behaviour alone. The 9 missing edges are predominantly short-range connections along the main stem
and at smaller tributary junctions, where the signal from upstream propagation is weaker relative to
the noise from local catchment variability. The 3 false positive edges, by contrast, connect
geographically distant stations and likely reflect residual co-exceedance not fully absorbed by the
proxy. The proxy adjustment suppresses many dense associations that would otherwise be induced by
region-wide storm events, while retaining the dominant propagation pathways.

\begin{figure}[htbp]
    \centering
    \includegraphics[width=0.98\textwidth]{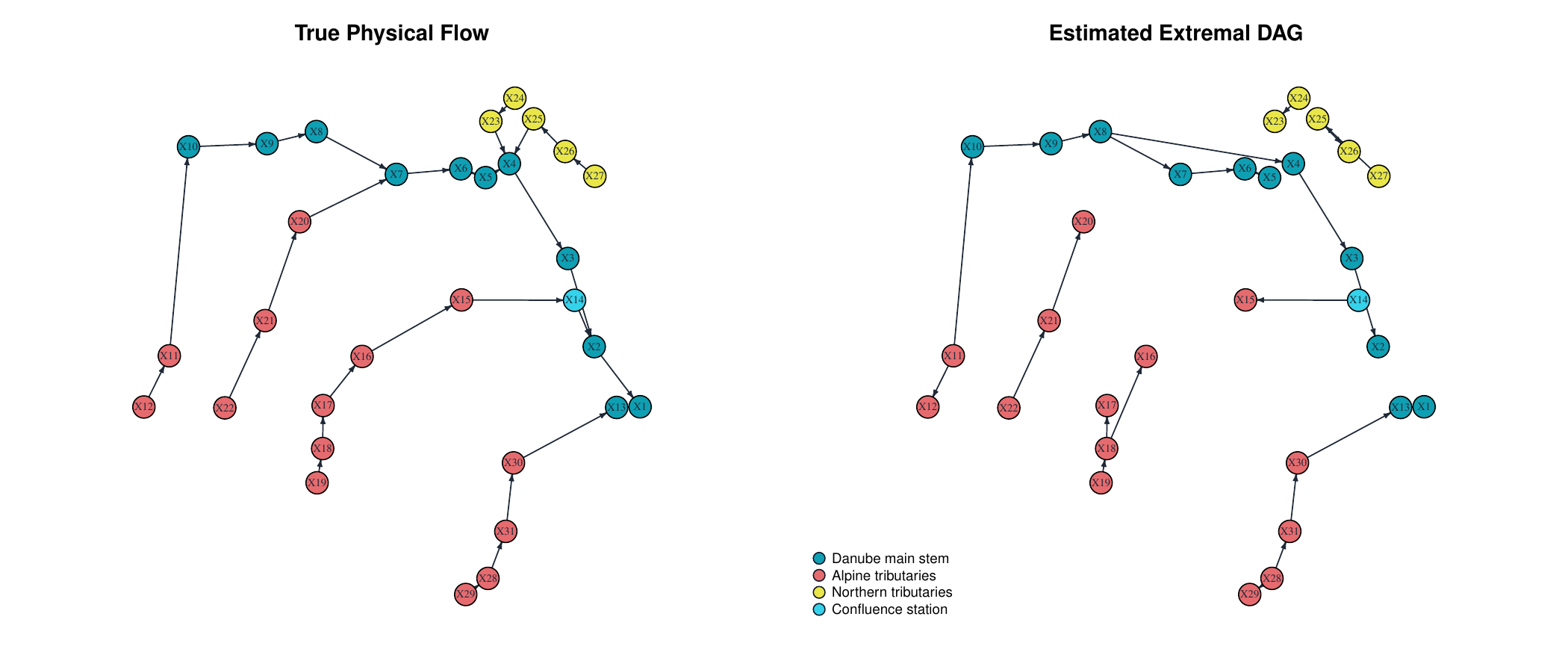}
    \caption{Upper Danube comparison. Left Panel shows the physical river-flow graph, while the
    right panel shows the estimated extremal DAG, both on the same geographic layout. Node colors
    indicate basin groups: Danube main stem (teal), Alpine tributaries (salmon), northern
    tributaries (yellow), and the confluence station (cyan).}
    \label{fig:danube_network}
\end{figure}

\subsection{S\&P 500 Tail Risk Network}
\label{sec:app2}

We next apply the method to daily equity returns from S\&P 500 constituents. The data are obtained from Stooq\footnote{\url{https://stooq.com}} over the period January 2005 to March 2025, yielding approximately $n=5{,}093$ trading
days and $p=103$ stocks spanning all eleven GICS sectors. Loss variables are defined by $L_{t,j}=-r_{t,j}$,
where $r_{t,j}$ denotes the log-return of stock $j$ on day $t$.

A main difficulty in this application is the presence of strong market-wide tail dependence during
periods of systemic stress. To attenuate this common component, we construct a three-dimensional
proxy
$\bm{S}_t = (\mathrm{VIX}_t,\,-r^{\mathrm{SP500}}_t,\,\Delta \mathrm{DGS10}_t)^\top,$
where $\mathrm{VIX}_t$ captures implied market volatility (a forward-looking measure of expected
tail risk), $-r^{\mathrm{SP500}}_t$ captures broad equity drawdowns (the contemporaneous market-wide
loss), and $\Delta \mathrm{DGS10}_t$ captures daily changes in the 10-year Treasury yield (a proxy
for flight-to-quality shifts that can induce simultaneous tail dependence across equity sectors).
We retain $k=4{,}120$ tail observations using the adaptive threshold
$q_{\mathrm{conf}} = 1 - n^{0.7}/n$; see Appendix~\ref{app:application_preprocessing} for data cleaning, preprocessing,
and tuning details. The skeleton is then estimated by proxy-adjusted nodewise extremal LASSO with
$\lambda = c\sqrt{\log(p+1)/k}$, using $c=5$, and edges are oriented with
$\gamma_{\mathrm{EBIC}}=10$ and maximum in-degree 10. The resulting graph contains 113 directed
edges among 103 nodes.

Figure~\ref{fig:sp500_dag} shows the estimated network. The graph exhibits clear within-sector
structure together with a smaller number of cross-sector transmission pathways. Utilities form the
most prominent cluster, with AEP (out-degree 5) and XEL (out-degree 5) appearing as major
transmitters and several downstream links concentrated within the same sector. Real-estate names
occupy a more intermediate position, connecting the Utilities cluster to parts of the Financial and
Consumer sectors. The estimated graph also identifies a technology subnetwork involving KLAC, AMAT,
LRCX, and NVDA, which is consistent with concentrated dependence within semiconductor-related supply
chains. More broadly, the hub structure, with a small number of high out-degree nodes concentrated
in Utilities, Financials, and Materials, suggests that tail risk propagation during market stress is
channelled through a few sector-specific transmitters rather than diffusing uniformly across the
market. These patterns should not be interpreted as definitive ground-truth causal claims in a
market without known validation structure. Rather, the main value of the application is that the
estimated graph is sparse, sectorally interpretable, and qualitatively aligned with plausible
channels of tail risk propagation under common market stress.

\begin{figure}[H]
    \centering
    \includegraphics[width=0.9\textwidth,trim=70 120 70 40,clip]{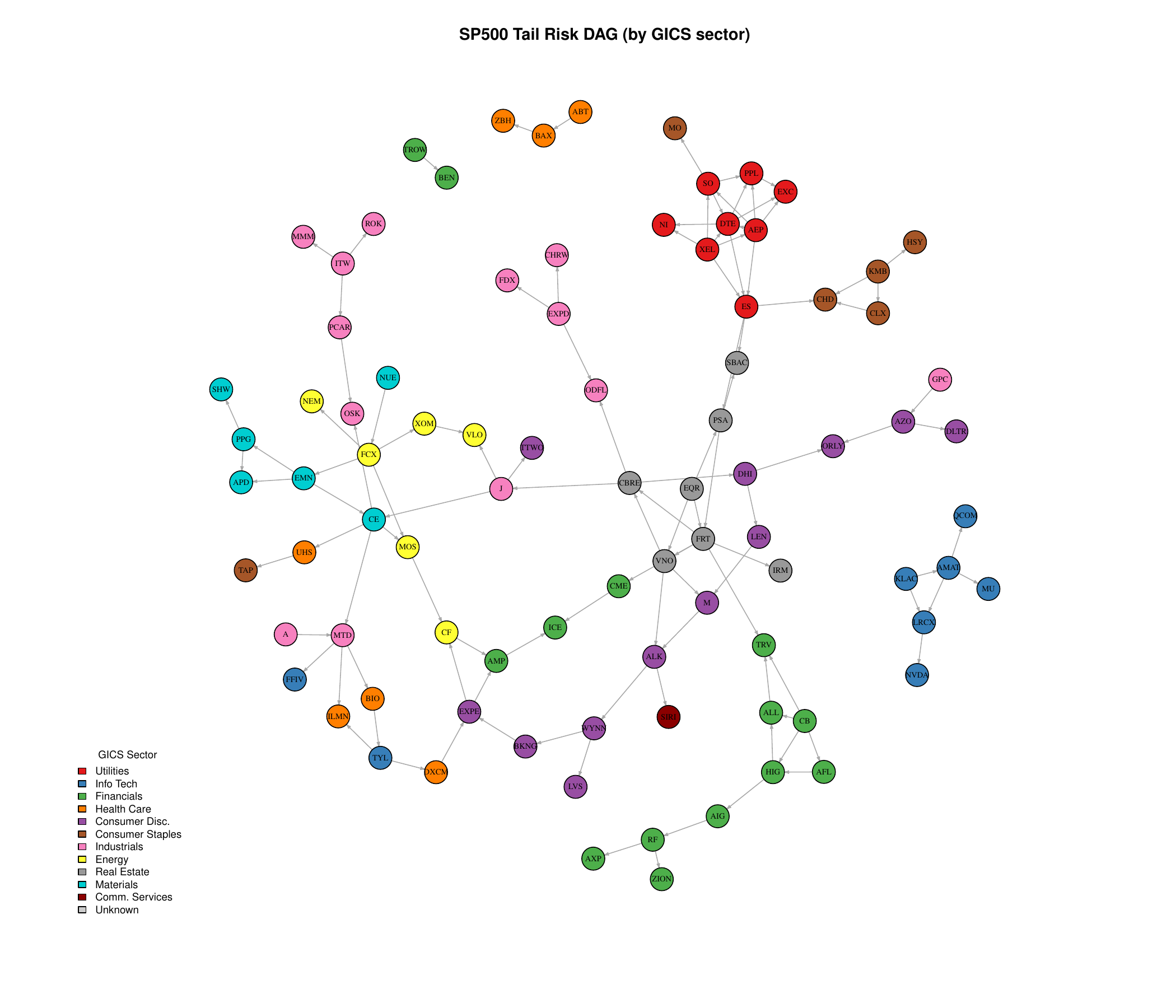}
    \caption{Estimated tail risk network among S\&P 500 constituents. Nodes are colored by GICS
    sector, and directed edges indicate inferred directions of tail risk propagation.}
    \label{fig:sp500_dag}
\end{figure}

\section{Conclusion}
\label{sec:discussion}

We have proposed a two-stage procedure for causal discovery in multivariate extremes, based on
tail-induced asymmetry. The first stage screens for a sparse candidate skeleton with proxy-adjusted
extremal nodewise regression. The second stage orients candidate edges by score minimisation under a
max-linear envelope with $\ell_1$ loss. The objective is not to recover bulk dependence, but to
extract directional information from the tail regime.

The evidence is consistent across theory, simulation, and data analysis. At population level,
Theorem~\ref{thm:asymmetry_identifiability} establishes forward--backward separation of tail
prediction error in a canonical bivariate max-linear setting. In finite samples, the simulation
results in Section~\ref{sec:simulation} show that proxy adjustment reduces false positives under
latent common shocks while retaining useful directional recovery. In the Danube application, where a
physical flow topology is available, 17 of 21 correctly identified skeleton edges are oriented
correctly (81\%). In the S\&P 500 analysis, where no ground-truth DAG is available, the estimated
network is interpreted as a sparse and sectorally coherent representation of tail-risk transmission,
rather than as validated causal truth.

The current guarantees are deliberately scoped. The identifiability argument is developed for a
bivariate max-linear model. The high-dimensional results rely on a working H\"usler--Reiss
approximation and a population score-separation condition. The superset non-improvement property is
verified for the max-linear subclass (Proposition~B.1, Appendix~\ref{app:proof_multivariate_identifiability}), not yet for the
full class in Assumption~\ref{assum:recursive_sem}. Moreover,
Theorem~\ref{thm:orientation_consistency} concerns an oracle global minimiser, whereas the
implemented orientation step is greedy. This theory--algorithm gap remains open. Proxy adjustment
should therefore be viewed as a structural safeguard against common-shock dependence, not as a full
correction for hidden confounding.

In terms of future work, directions appear most immediate. First, uncertainty quantification for
learned edges (for example, bootstrap stability or tail-focused subsampling) would improve practical
interpretability. Second, computational refinements are needed when the candidate skeleton is dense
and orientation search becomes costly. Third, extending orientation guarantees beyond the max-linear
envelope, together with richer domain-specific proxy design, would broaden applicability.

\appendix

\noindent\textbf{Additional high-dimensional scaling metrics.}
For completeness, Figure~\ref{fig:exp1_precision_recall_supp} reports the corresponding Precision
and Recall boxplots for the high-dimensional scaling experiment from Section~5.1 of the main paper.
These metrics complement the main-text $F_1$ and SHD summaries by showing how screening sensitivity
and positive predictive value change with dimension across methods.

\begin{figure}[htbp]
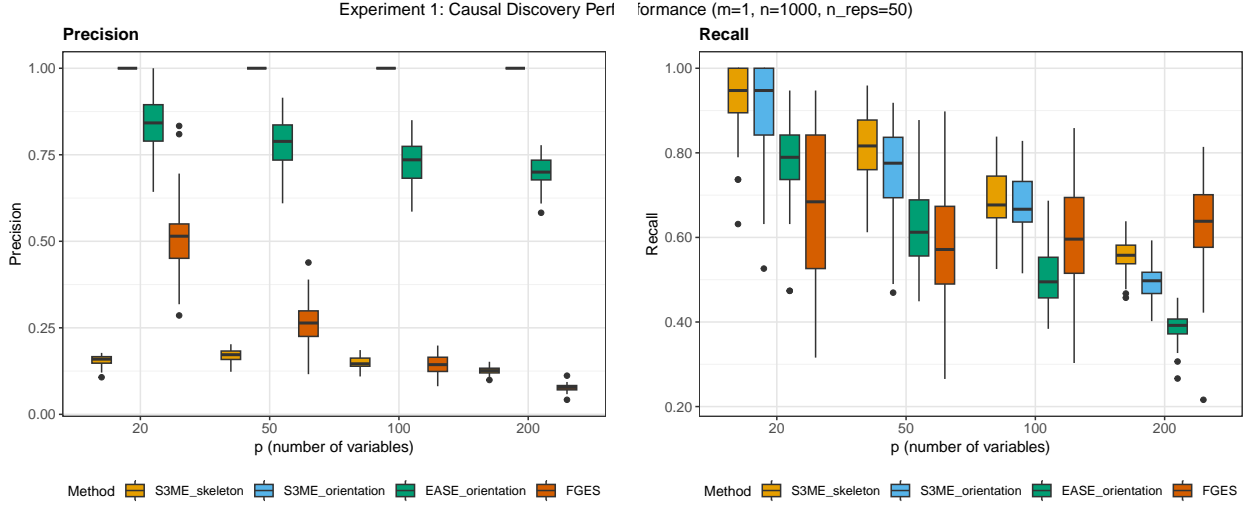

\centering
\includegraphics[width=0.49\textwidth,trim=0 360 432 0,clip]{Figures/exp1_metrics_boxplot.pdf}\hfill
\includegraphics[width=0.49\textwidth,trim=432 360 0 0,clip]{Figures/exp1_metrics_boxplot.pdf}
\caption{High-dimensional scaling results ($m=1$, $n=1000$, 50 replicates). Boxplots summarize
Precision (left) and Recall (right) across $p\in\{20,50,100,200\}$ for S3ME skeleton, EASE, FGES,
and S3ME.}
\label{fig:exp1_precision_recall_supp}
\end{figure}

% =============================================================
\section{Proof of Theorem~1 (Tail Asymmetry Identifiability)}
\label{app:proof_asymmetry_identifiability}
% =============================================================

\begin{proof}
We prove the forward and backward bounds separately. {Throughout the proof, $Y_j = 1/(1-F_j(X_j))$ denotes the Pareto-standardised variable, $Z_j = \log Y_j - \log u$ the log-tail coordinate, and $\tilde u$ the $X$-scale threshold corresponding to $u$ via the probability integral transform, so that $\{X_j > \tilde u\} \equiv \{Y_j > u\} \equiv \{Z_j > 0\}$. All conditioning events are equivalent up to monotone transformations.}

\paragraph{Forward direction ($1\to 2$).}
Condition on the tail event $\{Z_1>0\}$, which is equivalent to $\{Y_1>u\}$ and hence to
$\{X_1>\tilde u\}$ for a corresponding threshold $\tilde u\to\infty$ under the exact probability
integral transform. Under the max-linear specification,
\[
X_2=\max(cX_1,\epsilon_2).
\]
Since $\epsilon_2$ is independent of $X_1$ and $\epsilon_2\sim \text{Fr\'echet}(1)$, we have, for
$\tilde u$ sufficiently large,
\begin{align}
\mathbb{P}(\epsilon_2>cX_1\mid X_1>\tilde u)
&=\mathbb{E}\!\left[\mathbb{P}(\epsilon_2>cX_1\mid X_1)\,\big|\,X_1>\tilde u\right] \nonumber\\
&=\mathbb{E}\!\left[1-\exp\!\left(-\frac{1}{cX_1}\right)\Big|\,X_1>\tilde u\right] \nonumber\\
&\le \mathbb{E}\!\left[\frac{1}{cX_1}\Big|\,X_1>\tilde u\right]
\le \frac{1}{c\tilde u}\ \longrightarrow\ 0,
\end{align}
as $\tilde u\to\infty$. Therefore,
\[
\mathbb{P}(X_2=cX_1\mid X_1>\tilde u)\ \to\ 1.
\]
Since $\mathbb{P}(X_2=cX_1\mid X_1>\tilde u)\to 1$, we have $X_2/X_1 \to c$ in probability on
$\{X_1>\tilde u\}$. It remains to track the effect of the Pareto standardization $Y_j =
1/(1-F_j(X_j))$. For $X_1\sim\text{Fr\'echet}(1)$, $1-F_1(x)\sim 1/x$ as $x\to\infty$, so $Y_1 \sim
X_1$ uniformly on $\{X_1>\tilde u\}$. For $X_2=\max(cX_1,\epsilon_2)$, $\mathbb{P}(X_2>x)\sim
(c+1)/x$ and hence $1-F_2(x)\sim (c+1)/x$, so $Y_2 \sim X_2/(c+1)$ uniformly on $\{X_2>\tilde u\}$.
Combining these with $\mathbb{P}(X_2=cX_1\mid X_1>\tilde u)\to 1$ and by standard arguments for
conditional convergence in probability,
\[
\frac{Y_2}{Y_1}\ \to\ \frac{c}{c+1} \;=:\; \kappa
\qquad \text{in probability given } X_1>\tilde u.
\]
The log-tail coordinates $Z_j = \log Y_j - \log u$ then satisfy
\[
Z_2 - Z_1 \to \log \kappa
\quad \text{in probability given } Z_1 > 0.
\]
Choosing any finite $c_0<\log\kappa$ and the predictor
$\hat Z_2 = \max(Z_1 + \log\kappa,\, c_0) = \max(Z_1,\, c_0-\log\kappa) + \log\kappa$
yields $\hat Z_2 \in \mathcal H$ and $\hat Z_2 = Z_1 + \log\kappa$ on the event $\{Z_1>0\}$. Hence,
\[
\mathbb{E}\!\left[|Z_2-\hat Z_2|\mid Z_1>0\right]\ \to\ 0,
\]
so the optimal tail prediction risk satisfies $R^\star_{2\mid 1}(u) \to 0$ as $u \to \infty$.

\paragraph{Backward direction ($2\to 1$).}
Now condition on $\{X_2>\tilde u\}$ for $\tilde u\to\infty$. The exceedance $\{X_2>\tilde u\}$ can
occur via two asymptotically distinct mechanisms:

(i) propagation from an upstream extreme $cX_1$; or (ii) a local extreme innovation $\epsilon_2$
dominating $cX_1$.

Specifically, since $X_1$ and $\epsilon_2$ are independent Fr\'echet$(1)$ variables,
\[
\mathbb{P}(\epsilon_2 > u) \sim u^{-1}, \qquad \mathbb{P}(cX_1 > u) = \mathbb{P}(X_1 > u/c) \sim
c\,u^{-1}.
\]
By independence, the intersection term is $o(u^{-1})$, so
\[
\mathbb{P}(X_2 > u) = \mathbb{P}(\max(cX_1, \epsilon_2) > u) \sim (c+1)\,u^{-1}.
\]
Meanwhile, on the event $\{\epsilon_2 > cX_1,\; X_2 > u\}$, we have $X_2 = \epsilon_2$, and
\[
\mathbb{P}(\epsilon_2 > cX_1,\; X_2 > u) = \mathbb{P}(\epsilon_2 > u,\; \epsilon_2 > cX_1) \sim
\mathbb{P}(\epsilon_2 > u) \sim u^{-1}.
\]

Write $A_u := \{X_2 > \tilde u\}$ and $B := \{\epsilon_2 > cX_1\}$. The calculation above gives
\begin{equation}
\lim_{u\to\infty}\mathbb{P}(B\mid A_u)
= \frac{1}{c+1}.
\end{equation}

Since node~1 is a source, $X_1=\epsilon_1$, so on $A_u \cap B$ we have $X_2=\epsilon_2$ and
$\epsilon_1 < \epsilon_2/c$. Any predictor $h\in\mathcal H$ can be written as
$h(z)=\max(z,\gamma)+d$; equivalently, setting $c_{2\to 1}=d$ and $c_0'=\gamma+d$ gives the
reparametrization $h(z)=\max(z+c_{2\to 1},\,c_0')$, which we use below. For any max-linear envelope
predictor $\hat Z_1 = \max(Z_2 + c_{2 \to 1}, c_0')$, by the law of total expectation over $B$ and
$B^c$ within $A_u$, and dropping the non-negative $B^c$ term, we have $\mathbb{E}[|Z_1-\hat Z_1|\mid
A_u]\ge \mathbb{P}(B\mid A_u)\,\mathbb{E}[|Z_1-\hat Z_1|\mid A_u, B]$ for every predictor. Since
$\mathbb{P}(B\mid A_u)$ is free of the predictor parameters, taking the infimum over $(c_{2\to
1},c_0')$ on both sides yields
\[
\inf_{c_{2\to 1},\,c_0'}\mathbb{E}\!\left[|Z_1-\hat Z_1|\mid A_u\right]
\;\ge\; \mathbb{P}(B\mid A_u)\,\inf_{c_{2\to 1},\,c_0'}\mathbb{E}\!\left[|Z_1-\max(Z_2+c_{2\to
1},c_0')|\mid A_u, B\right].
\]

We now lower-bound the inner infimum on $A_u \cap B$. On this event, $X_2 = \epsilon_2$ and $X_1 =
\epsilon_1$ with $\epsilon_1 \perp \epsilon_2$. Given $\epsilon_2 = y$ and $\epsilon_2 >
c\epsilon_1$, the independence of $\epsilon_1$ and $\epsilon_2$ yields, for $x < y/c$,
\[
\Pr(\epsilon_1 \le x \mid \epsilon_2 = y,\; \epsilon_2 > c\epsilon_1)
= \frac{\Pr(\epsilon_1 \le x,\; \epsilon_1 < y/c)}{\Pr(\epsilon_1 < y/c)}
= \frac{F_1(x)}{F_1(y/c)}.
\]
Since $y \asymp u e^z \to \infty$ on $\{Z_2 = z,\, A_u \cap B\}$, we have $F_1(y/c) \to 1$, and
hence for every fixed $x$,
\[
\Pr(\epsilon_1 \le x \mid Z_2 = z,\, A_u \cap B) \to F_1(x).
\]
That is, the conditional distribution of $\epsilon_1$ given $(Z_2, A_u, B)$ converges pointwise to
the unconditional $\mathrm{Fr\acute{e}chet}(1)${, for each fixed $x$, as $u\to\infty$}.

Since $Y_1 = 1/(1-F_1(\epsilon_1))$ is the exact probability-integral-transform of $\epsilon_1$ to
unit Pareto, $\log Y_1 \sim \mathrm{Exp}(1)$ exactly, and $Z_1 = \log Y_1 - \log u$. By the
convergence above, for any $h\in\mathcal{H}$ and any $z$,
\[
\mathbb{E}\!\left[|Z_1 - h(z)|\mid Z_2 = z,\, A_u \cap B\right]
= \mathbb{E}\!\left[|\log Y_1 - \log u - h(z)|\right] + o(1)
\;\ge\; \mathbb{E}\!\left[|\log Y_1 - \log 2|\right] + o(1),
\]
where the inequality uses the fact that $\log 2 - \log u$ is the median of $Z_1 = \log Y_1 - \log
u$, and the median minimizes the $\ell_1$ risk. Since $\mathbb{E}[|\mathrm{Exp}(1) - \log 2|] = \log
2$, we have
\[
\liminf_{u\to\infty}\inf_{h\in\mathcal{H}}\mathbb{E}\!\left[|Z_1-h(Z_2)|\mid A_u \cap B\right]
\;\ge\; \log 2 \;>\; 0.
\]
Combining the two bounds gives
\[
\liminf_{u\to\infty} R^\star_{1\mid 2}(u)
\ge \left(\lim_{u\to\infty}\mathbb{P}(B\mid A_u)\right)
\left(\liminf_{u\to\infty}\inf_{h\in\mathcal{H}}\mathbb{E}\!\left[|Z_1-h(Z_2)|\mid A_u \cap
B\right]\right)
\ge \frac{\log 2}{c+1} > 0.
\]

\paragraph{From tail risk separation to identifiability.}
Finally, for the $\ell_1$-based max-linear envelope scoring criterion, the nodewise contribution is
monotone in the (tail-weighted) absolute prediction error. Since the forward tail risk converges to
zero while the reverse tail risk has a strictly positive $\liminf$ as $u \to \infty$, any monotone
$\ell_1$-based score evaluated at a sufficiently large threshold $u$ strictly prefers the true
direction $1\to 2$ over $2\to 1$, yielding asymptotic identifiability.
\end{proof}

% =============================================================
\section{Proof of Theorem~2 (Nodewise Identifiability)}
\label{app:proof_multivariate_identifiability}
% =============================================================

This section provides the technical verification that the abstract population-separation conditions
of Section~4 are satisfied by the specific EBIC-based score used in our procedure.

\begin{proposition}[Superset non-improvement under the max-linear benchmark]
\label{prop:penalty_domination}
Under the fully observed recursive max-linear specification of Assumption~3
with $f_j(\{x_m\}) = \max_{m\in\mathrm{pa}(j)} a_{mj}x_m$, unit Fr\'echet innovations, and no latent
confounders, for any fixed node $j$ and any superset $A\supseteq \mathrm{pa}^*(j)$,
\[
R^*_{j\mid A}(u)=R^*_{j\mid\mathrm{pa}^*(j)}(u).
\]
In particular, adding spurious parents does not improve population tail risk.
\end{proposition}

\begin{proof}
Fix node $j$ and write $A=\mathrm{pa}^*(j)\cup B$ with $B\cap\mathrm{pa}^*(j)=\emptyset$. Under the
recursive max-linear model and Proposition~\ref{prop:bounded_approx_error}, compare the log-sum-exp
proxies under the true parent set and superset:
\[
T_0 := \sum_{m\in\mathrm{pa}^*(j)} a_{mj}e^{L_m} + \epsilon_j,
\qquad
T_B(c_B) := \sum_{m\in B} e^{c_{m\to j}+L_m},
\]
\[
\tilde L_j^{A}(c_B)=\log\!\big(T_0+T_B(c_B)\big),
\qquad
\tilde L_j^{\mathrm{pa}}=\log(T_0).
\]
Since $T_B(c_B)\ge0$, we have $\tilde L_j^{A}(c_B)\ge \tilde L_j^{\mathrm{pa}}\ge L_j$ pointwise, so
\[
|L_j-\tilde L_j^{A}(c_B)|\ge |L_j-\tilde L_j^{\mathrm{pa}}|.
\]
Taking expectations and infimum over finite $c_B$ yields
\[
R^*_{j\mid A}(u)\ge R^*_{j\mid\mathrm{pa}^*(j)}(u).
\]
Conversely, set $c_{m\to j}^{(r)}=-r$ for $m\in B$. Then $T_B(c_B^{(r)})\downarrow0$ and $\tilde
L_j^{A}(c_B^{(r)})\downarrow\tilde L_j^{\mathrm{pa}}$ pointwise as $r\to\infty$. By monotone
convergence of nonnegative errors,
\[
\lim_{r\to\infty} R_{j\mid A}(u;c_B^{(r)}) = R^*_{j\mid\mathrm{pa}^*(j)}(u),
\]
which implies $R^*_{j\mid A}(u)\le R^*_{j\mid\mathrm{pa}^*(j)}(u)$. Therefore equality holds.
\end{proof}

\begin{proof}[Proof of Theorem~2]
Fix $A\subseteq \mathrm{nb}^*(j)$ with $A\neq \mathrm{pa}^*(j)$. There are two cases.

\textit{Case 1: $A\not\supseteq \mathrm{pa}^*(j)$.} By Assumption~5, there
exist constants $\Delta_{\mathrm{miss},j}>0$ and $u_{0,j}$ such that for all $u\ge u_{0,j}$,
\[
R^*_{j\mid A}(u)-R^*_{j\mid\mathrm{pa}^*(j)}(u)\ge \Delta_{\mathrm{miss},j}>0.
\]
Hence any parent set omitting at least one true parent has strictly larger population risk.

\textit{Case 2: $A\supsetneq \mathrm{pa}^*(j)$.} By Assumption~6,
\[
R^*_{j\mid A}(u)\ge R^*_{j\mid\mathrm{pa}^*(j)}(u)
\]
for sufficiently large $u$.

Thus $\mathrm{pa}^*(j)$ is a minimizer of $R^*_{j\mid A}(u)$. Moreover, any minimizer cannot be in
Case 1, so every minimizer must contain $\mathrm{pa}^*(j)$. Therefore $\mathrm{pa}^*(j)$ is the
unique inclusion-minimal minimizer.
\end{proof}

\begin{proof}[Proof of Corollary~1]
For any skeleton-compatible DAG $G$, Theorem~2 gives nodewise
inequalities
\[
R^*_{j\mid\mathrm{pa}_G(j)}(u)\ge R^*_{j\mid\mathrm{pa}^*(j)}(u),\qquad j=1,\dots,p.
\]
Summing over $j$ yields
\[
R^*(G;u)\ge R^*(\mathcal G^*;u),
\]
so $\mathcal G^*$ is a global minimizer. If $G$ is any global minimizer, equality must hold at every
node, hence each $\mathrm{pa}_G(j)$ is a nodewise minimizer and must contain $\mathrm{pa}^*(j)$ by
Theorem~2. Therefore every global minimizer contains $\mathcal G^*$
edgewise, and $\mathcal G^*$ is the unique inclusion-minimal global minimizer.
\end{proof}

% =============================================================
\section{Uniformly Bounded Approximation Error}
\label{app:proof_prop3}
% =============================================================

\noindent\textbf{Precise regularity conditions for Assumption~5.}
The compact statement of Assumption~5 in the main paper abbreviates the following rate conditions
for the true undirected skeleton, the H\"usler--Reiss limit, and the proxy-adjusted model:
\begin{enumerate}
\item \textbf{(Sparsity and Intermediate Effective Sample Size)} For each node $j$, let
$\mathcal{N}_{j}=\{i:\{i,j\}\in E^{*}_{\mathrm{skel}}\}$ and $s_{j}=|\mathcal{N}_{j}|$ denote its
neighbor set and degree in the skeleton, and let $s_{\max}=\max_{j}s_{j}$. We assume
$s_{\max}=o\!\left(\sqrt{\frac{k}{\log p}}\right)$.

\item \textbf{(Bounded Eigenvalues)} For each node $j$ with true neighbour set $\mathcal{N}_j$, the submatrix
$\Sigma_{\mathcal{N}_j\mathcal{N}_j}$ of the H\"usler--Reiss covariance matrix
$\Sigma=\Theta^{-1}$ satisfies $\Lambda_{\min}(\Sigma_{\mathcal{N}_j\mathcal{N}_j})\ge
C_{\min}$ for some absolute constant $C_{\min}>0$, where $\Lambda_{\min}(\cdot)$ denotes
the smallest eigenvalue.

\item \textbf{(Minimum Signal Strength in the Tail)} The non-zero regression coefficients in the
tail limit are sufficiently strong to be detected at effective sample size $k$. Specifically, for
some sufficiently large constant $C>0$, $\min_{j}\min_{i\in \mathcal{N}_{j}}|\beta_{ji}^{*}| \ge
\frac{C}{\sqrt{C_{\min}}}\sqrt{\frac{s_{\max}\log p}{k}}$.

\item \textbf{(Score Bias Control)} For each nodewise regression after proxy partialling-out, the
population score contribution of the log-sum-exp residual bias satisfies
$\bigl\|\mathbb{E}[X^{\mathrm{res}}\tilde{\Delta}]\bigr\|_\infty \le c_\Delta \lambda$ for some
constant $c_\Delta>0$, where $X^{\mathrm{res}}$ denotes the residualized design and $\tilde{\Delta} =
\Delta_{\mathrm{model}} - \bar{\Delta}\mathbf{1}_k$ is the mean-centered approximation
bias from replacing the max-linear structural equation by a log-sum-exp proxy, bounded
by $\log(|\mathrm{pa}(j)|+1)$ per Proposition~C.1 below.
\end{enumerate}

\begin{proposition}[Uniformly bounded approximation error]
\label{prop:bounded_approx_error}
Let $\mathbf{X} = (X_1,\dots,X_p)$ follow a recursive max-linear SEM under Assumption~3 with
$f_j(\{x_m\})=\max_{m\in\mathrm{pa}(j)}a_{mj}x_m$ (edge weights $a_{mj}>0$) and unit Fr\'echet
innovations. Write $L_j = \log X_j$. The structural equation for node $j$ takes the form
\[
L_j \;=\; \bigvee_{m\in \mathrm{pa}(j)}\bigl(L_m+\log a_{mj}\bigr)\;\vee\;\log\epsilon_j .
\]
Define the Log-Sum-Exp approximation
\[
\tilde L_j \;=\; \log\!\left(\sum_{m\in \mathrm{pa}(j)} a_{mj}\,e^{L_m} + \epsilon_j\right),
\]
and the approximation error $\mathcal{E}_j=\tilde L_j-L_j$. Then, for any observation,
$\mathcal{E}_j$ is deterministically bounded:
\[
0 \le \mathcal{E}_j \le \log\!\big(|\mathrm{pa}(j)| + 1\big).
\]
Consequently, the log-sum-exp proxy is equivalent to the max-linear structural equation up to a
strictly bounded $O(1)$ constant, independent of the tail threshold $u$.
\end{proposition}

\begin{proof}
Let $X_{(1)} = \max_{m \in \mathrm{pa}(j)} (a_{mj}e^{L_m}) \vee \epsilon_j$ denote the maximum term
in the true structural equation. The approximation error induced by the Log-Sum-Exp mapping is
exactly
\[
\mathcal{E}_j = \log \left( \sum_{m \in \mathrm{pa}(j)} a_{mj}e^{L_m} + \epsilon_j \right) - \log
X_{(1)} = \log \left( 1 + \sum_{m \neq m^*} \frac{X_{m}}{X_{(1)}} \right),
\]
where the sum inside the logarithm contains exactly $|\mathrm{pa}(j)|$ non-maximum terms. By the
definition of the maximum, the ratio $X_{m}/X_{(1)} \le 1$ for all terms. Therefore, we
deterministically obtain the uniform bound $0 \le \mathcal{E}_j \le \log(1 + |\mathrm{pa}(j)|)$.
\end{proof}

% =============================================================
\section{Proof of Lemma~1 (Uniform Score Concentration)}
\label{app:proof_uniform_score}
% =============================================================

\begin{proof}[Proof of Lemma~1]
By Theorem~3, condition on the event $E^*_{\mathrm{skel}} \subseteq
\hat{E}_{\mathrm{skel}}$ with $|\hat{E}_{\mathrm{skel}}| \le c_3 p s_{\max}$. For each node $j$, the
parent set $A$ must satisfy $A \subseteq \mathrm{nb}(j;\hat{E}_{\mathrm{skel}})$ and $|A| \le
s_{\max}$. The number of such sets is at most
\[
\sum_{m=0}^{s_{\max}} \binom{|\mathrm{nb}(j;\hat{E}_{\mathrm{skel}})|}{m} \le \sum_{m=0}^{s_{\max}}
p^m \le (s_{\max}+1)\,p^{s_{\max}}.
\]
Union-bounding the assumed per-set concentration over all $A$ for fixed $j$,
\[
\mathbb{P}\Bigl(\sup_{A}\bigl|\hat{R}_{j\mid A}(u)-R^*_{j\mid A}(u)\bigr|>\epsilon\Bigr) \le
c_1\exp\bigl(s_{\max}\log p + \log(s_{\max}+1) - c_2 k\min(\epsilon^2,\epsilon)\bigr).
\]
A further union bound over $p$ nodes gives
\[
\mathbb{P}\Bigl(\max_{1\le j\le p}\sup_{A}\bigl|\hat{R}_{j\mid A}(u)-R^*_{j\mid
A}(u)\bigr|>\epsilon\Bigr) \le c_1\exp\bigl((s_{\max}+1)\log p + \log(s_{\max}+1) - c_2
k\min(\epsilon^2,\epsilon)\bigr).
\]
Under $s_{\max}=o(\sqrt{k/\log p})$ and $\log p=o(k)$, the exponent diverges to $-\infty$ when
$\epsilon = C\sqrt{(s_{\max}+1)\log p/k}$ for sufficiently large $C$, giving the stated rate.
\end{proof}

% =============================================================
\section{Proof of Theorem~3 (Skeleton Consistency)}
\label{app:proof_skeleton_consistency}
% =============================================================

\begin{proof}[Proof of Theorem~3]
Fix a target node $j \in V$ and, to simplify notation, drop the subscript $j$. Because the proxy
enters the nodewise regression as an unpenalized regressor, we first partial out its contribution
and work with the residualized response and residualized design. To distinguish from the
Pareto-standardized variables $Y$ and the exceedance data $\{Z^{(i)}\}$ in
Section~3, we write $\mathbf{y} \in \mathbb{R}^k$ for the resulting response vector
of proxy-adjusted log-exceedances and $\mathbf{X} \in \mathbb{R}^{k \times (p-1)}$ for the
corresponding design matrix of residualized predictors. Let $\beta^*$ be the true sparse regression
coefficient vector associated with the extremal precision matrix. Let $\mathcal{S} =
\mathrm{supp}(\beta^*)$ be the true support set and $|\mathcal{S}| = s$.

Under Assumption~4 and Proposition~\ref{prop:bounded_approx_error}, the proxy-adjusted
tail-domain model can be written as $\mathbf{y} = \mathbf{X}\beta^* + W + \Delta_{\mathrm{model}}$,
where the stochastic noise $W$ captures the H\"usler--Reiss approximation error and
$\Delta_{\mathrm{model}}$ is the log-sum-exp bias, deterministically bounded by
$\log(|\mathrm{pa}(j)|+1)$. By Proposition~\ref{prop:bounded_approx_error}, each component of
$\Delta_{\mathrm{model}}$ is bounded by $\log(|\mathrm{pa}(j)|+1)$. The nodewise Lasso includes an
unpenalized intercept, which absorbs the sample mean of $\Delta_{\mathrm{model}}$ via the
first-order KKT condition. The residual bias $\tilde{\Delta} = \Delta_{\mathrm{model}} -
\bar{\Delta}\,\mathbf{1}_k$ is therefore zero-mean and bounded. Under the extremal precision
characterization, $\mathrm{supp}(\beta^*)$ coincides with the true neighbourhood of node $j$.

The residual bias $\tilde{\Delta}=\Delta_{\mathrm{model}}-\bar{\Delta}\mathbf{1}_k$ is centered, but
its score contribution need not have zero population mean. We therefore write
\[
\frac{1}{k}\mathbf{X}^\top\tilde{\Delta}
=\mathbb{E}[X^{\mathrm{res}}\tilde{\Delta}]+\left(\frac{1}{k}\mathbf{X}^\top\tilde{\Delta}-\mathbb{E}[X^{\mathrm{res}}\tilde{\Delta}]\right).
\]
By condition (4) of Assumption~8,
$\|\mathbb{E}[X^{\mathrm{res}}\tilde{\Delta}]\|_\infty \le c_\Delta\lambda$. For the fluctuation
term, note that in the log-tail domain the coordinates are exponential-type; under
Assumption~4, each coordinate of $X^{\mathrm{res}}$ is therefore sub-exponential
(equivalently, has uniformly bounded $\psi_1$-Orlicz norm). Since $\tilde{\Delta}$ is bounded
(Proposition~\ref{prop:bounded_approx_error}), the products
\[
U_{t\ell}:=X^{\mathrm{res}}_{t\ell}\tilde{\Delta}_t-\mathbb{E}[X^{\mathrm{res}}_{t\ell}\tilde{\Delta}_t]
\]
are centered sub-exponential with $\|U_{t\ell}\|_{\psi_1}\le K_U$ uniformly in $\ell$. Bernstein's
inequality for sub-exponential sums (e.g., \citealp[Corollary~2.8.3]{wainwright2019high}) yields,
for any $x>0$,
\[
\mathbb{P}\!\left(\left|\frac{1}{k}\sum_{t=1}^k U_{t\ell}\right|>x\right)
\le 2\exp\!\left[-c k\min\!\left(\frac{x^2}{K_U^2},\frac{x}{K_U}\right)\right].
\]
Applying a union bound over $\ell=1,\dots,p$ and taking $x=C\sqrt{\log p/k}$, with $\log p=o(k)$ so
that $x\to0$ and the quadratic regime applies for large $k$, we obtain
\[
\left\|\frac{1}{k}\mathbf{X}^\top\tilde{\Delta}-\mathbb{E}[X^{\mathrm{res}}\tilde{\Delta}]\right\|_\infty=O_p\!\left(\sqrt{\frac{\log
p}{k}}\right)=O_p(\lambda).
\]
Hence $\|\frac{1}{k}\mathbf{X}^\top\tilde{\Delta}\|_\infty=O_p(\lambda)$, which is absorbed into the
regularization parameter by enlarging the constant $C$ in the choice of $\lambda$.

\paragraph{Retaining true edges (sure screening).}
We show that for each $i \in \mathcal{S}$ (each true neighbour), $\hat{\beta}_i \neq 0$ with high
probability. By the standard $\ell_1$ error bound for the Lasso under the restricted eigenvalue
condition \citep{buhlmann2011statistics}, Assumption~8(2) gives
\begin{equation}
    \|\hat{\beta} - \beta^*\|_2 \le \frac{\sqrt{s}}{\kappa}\,\lambda
\end{equation}
with probability at least $1 - c_5 \exp(-c_6 k \lambda^2)$, for a constant $\kappa$ depending on
$C_{\min}$ through the restricted eigenvalue condition (Assumption~8(2)). In particular,
\begin{equation}
    \|\hat{\beta}_{\mathcal{S}} - \beta^*_{\mathcal{S}}\|_\infty \le \|\hat{\beta} - \beta^*\|_2 \le \frac{\sqrt{s}}{\kappa}\,\lambda = \mathcal{O}_p\!\left( \frac{1}{\sqrt{C_{\min}}} \sqrt{\frac{s_{\max}\log p}{k}} \right).
\end{equation}
By Assumption~8(3), $\min_{i \in \mathcal{S}} |\beta^*_i| \ge
\frac{C}{\sqrt{C_{\min}}} \sqrt{\frac{s_{\max}\log p}{k}}$ for a sufficiently large constant $C$, so
$|\hat{\beta}_i - \beta^*_i| < |\beta^*_i|$ and hence $\hat{\beta}_i \neq 0$ for each $i \in
\mathcal{S}$. Taking a union bound over the at most $s_{\max}$ true neighbours per node and over all
$p$ nodes, all true edges are retained with probability at least $1 - C_1 p s_{\max} \exp(-C_2 k
\lambda^2)$, which tends to one under the stated scaling regime.

\paragraph{Bounding the number of selected edges.}
The skeleton is constructed from the thresholded support $\hat{\mathcal{S}}^{(\tau)} = \{i :
|\hat{\beta}_i| > \tau\}$ with $\tau = \alpha \lambda$ for some $\alpha \in (0,1)$, so that any
variable with coefficient magnitude below the threshold is excluded. Since $\|\hat{\beta} -
\beta^*\|_2 \ge \sqrt{\sum_{i \in \hat{\mathcal{S}}^{(\tau)} \setminus \mathcal{S}}
|\hat{\beta}_i|^2} \ge \sqrt{|\hat{\mathcal{S}}^{(\tau)} \setminus \mathcal{S}|}\cdot\tau$, the
thresholding condition $|\hat{\beta}_i| > \tau$ for $i \in \hat{\mathcal{S}}^{(\tau)} \setminus
\mathcal{S}$ gives $\sqrt{|\hat{\mathcal{S}}^{(\tau)} \setminus \mathcal{S}|}\cdot\tau <
\frac{\sqrt{s}}{\kappa}\lambda$, hence $|\hat{\mathcal{S}}^{(\tau)}| \le (\frac{1}{\kappa\alpha} +
1) s$. Applying this bound to each of the $p$ nodewise regressions and taking the union yields
$|\hat{E}(\lambda)| \le c_3 \, p \, s_{\max}$ with high probability, for a constant $c_3$.

Combining the two parts, we obtain $\mathbb{P}(E^*_{\mathrm{skel}} \subseteq \hat{E}(\lambda)) \ge 1
- C_1 p s_{\max} \exp(-C_2 k \lambda^2)$ and, on this event, $|\hat{E}(\lambda)| \le c_3 \, p \,
s_{\max}$.
\end{proof}

% =============================================================
\section{Proof of Theorem~4 (Orientation Consistency)}
\label{app:proof_orientation_consistency}
% =============================================================

\begin{proof}[Proof of Theorem~4]
Let $\mathcal{G}_{c_3}$ denote the class of DAGs whose skeleton is a subgraph of some edge set $E$
with $E\supseteq E^{*}_{\mathrm{skel}}$, $|E|\le c_{3}ps_{\max}$, and nodewise in-degree at most
$s_{\max}$. By Theorem~3, on an event with probability tending to one,
\[
E^*_{\mathrm{skel}}\subseteq \hat{E}_{\mathrm{skel}},\qquad |\hat{E}_{\mathrm{skel}}|\le
c_3ps_{\max},
\]
so $\mathcal{G}(\hat{E}_{\mathrm{skel}})\subseteq\mathcal{G}_{c_3}$ and
$G^*\in\mathcal{G}(\hat{E}_{\mathrm{skel}})$.

Fix any $G\neq G^*$ in $\mathcal{G}(\hat{E}_{\mathrm{skel}})$ and define the set of changed nodes
\[
\mathcal{J}(G,G^*) := \bigl\{j:\ \mathrm{pa}_G(j)\neq \mathrm{pa}_{G^*}(j)\bigr\}.
\]
Partition $\mathcal{J}$ into
\begin{align*}
\mathcal{J}_{\mathrm{under}}(G,G^*) &:= \bigl\{j\in\mathcal{J}:
 \mathrm{pa}_G(j)\not\supseteq \mathrm{pa}^*(j)\bigr\}, \\
\mathcal{J}_{\mathrm{over}}(G,G^*) &:= \bigl\{j\in\mathcal{J}:
 \mathrm{pa}_G(j)\supsetneq \mathrm{pa}^*(j)\bigr\},
\end{align*}
noting that $\mathcal{J}=\mathcal{J}_{\mathrm{under}}\cup\mathcal{J}_{\mathrm{over}}$ (disjoint).
Write $A_j=\mathrm{pa}_G(j)$ and $A_j^*=\mathrm{pa}^*(j)$. The empirical score decomposes as
\[
\hat{\mathrm{Score}}_j(A) = \hat{f}_j(A) + \mathrm{pen}(|A|),
\]
where $\hat{f}_j(A)=\frac{k}{2}\log\frac{\mathrm{SAE}_j(A)}{k}$ is the empirical fit and
$\mathrm{pen}(m)=\frac{1}{2}(\log k+2\gamma_{\mathrm{EBIC}}\log p)\,m$ is the EBIC complexity
penalty. The global score difference is
\[
\hat{S}(G)-\hat{S}(G^*)=\sum_{j\in\mathcal{J}}\bigl[\hat{f}_j(A_j)-\hat{f}_j(A_j^*)+\mathrm{pen}(|A_j|)-\mathrm{pen}(|A_j^*|)\bigr].
\]
By Lemma~1,
\[
\eta_k:=\max_{1\le j\le p}\sup_{A}\bigl|\hat{R}_{j\mid A}(u)-R^*_{j\mid
A}(u)\bigr|=O_p\!\left(\sqrt{\frac{(s_{\max}+1)\log p}{k}}\right)=o_p(1).
\]
We show that each summand is strictly positive with probability tending to one.

\medskip
\noindent\textbf{Case~I ($j\in\mathcal{J}_{\mathrm{under}}$): $A_j\not\supseteq A_j^*$.}
By assumption~(i), $R^*_{j\mid A_j}(u)-R^*_{j\mid A_j^*}(u)\ge\Delta_{\mathrm{miss}}$. A mean-value
expansion of $\hat{f}_j(A)=\frac{k}{2}\log\hat{R}_{j\mid A}(u)${ (possible since $\log$ is Lipschitz on $[r_{\min}(u)/2,\,2r_{\max}(u)]$),} gives
\[
\hat{f}_j(A_j)-\hat{f}_j(A_j^*)
= \frac{k}{2\,\bar R_j}\bigl[\hat{R}_{j\mid A_j}(u)-\hat{R}_{j\mid A_j^*}(u)\bigr]
\ge \frac{k}{2\,\bar R_j}\bigl[\Delta_{\mathrm{miss}}-2\eta_k\bigr],
\]
where $\bar R_j$ lies between $\hat R_{j\mid A_j}(u)$ and $\hat R_{j\mid A_j^*}(u)$. On the event
$\eta_k<\min\{r_{\min}(u)/2,\,r_{\max}(u)\}$, all admissible empirical risks lie in
$[r_{\min}(u)/2,\,2r_{\max}(u)]$, hence $\bar R_j\le 2r_{\max}(u)$ and therefore
\[
\hat{f}_j(A_j)-\hat{f}_j(A_j^*)
\ge \frac{k}{4r_{\max}(u)}\bigl[\Delta_{\mathrm{miss}}-2\eta_k\bigr].
\]
The penalty difference satisfies
\[
\bigl|\mathrm{pen}(|A_j|)-\mathrm{pen}(|A_j^*|)\bigr|\le \tfrac{1}{2}(\log
k+2\gamma_{\mathrm{EBIC}}\log p)\,s_{\max}=O(\log k\cdot s_{\max}).
\]
Since $\frac{k\Delta_{\mathrm{miss}}}{4r_{\max}(u)}=\Omega(k)$ dominates $O(\log k\cdot s_{\max})$
for large $k$, the net score contribution is positive with probability tending to one.

\medskip
\noindent\textbf{Case~II ($j\in\mathcal{J}_{\mathrm{over}}$): $A_j\supsetneq A_j^*$.}
By assumption~(ii), $R^*_{j\mid A_j}(u)=R^*_{j\mid A_j^*}(u)$. Let $B_j:=A_j\setminus A_j^*$, so
$|B_j|\ge1$. By Assumption~6, uniformly over such equal-risk supersets,
\[
\hat f_j(A_j^*)-\hat f_j(A_j)=O_p(|B_j|).
\]
The EBIC penalty gap is
\[
\mathrm{pen}(|A_j|)-\mathrm{pen}(|A_j^*|)
=\tfrac{1}{2}(\log k+2\gamma_{\mathrm{EBIC}}\log p)\,|B_j|.
\]
Hence
\[
\hat{\mathrm{Score}}_j(A_j)-\hat{\mathrm{Score}}_j(A_j^*)
\ge -O_p(|B_j|)+\tfrac{1}{2}(\log k+2\gamma_{\mathrm{EBIC}}\log p)\,|B_j|.
\]
Since $\log k+2\gamma_{\mathrm{EBIC}}\log p\to\infty$, the penalty term dominates and the net score
contribution is positive with probability tending to one.

\medskip
\noindent\textbf{Graph-level conclusion.}
Both cases yield strictly positive per-node contributions with probability tending to one. Hence
\[
\hat{S}(G)-\hat{S}(G^*)>0
\]
for every $G\neq G^*$ with probability tending to one. Since
$G^*\in\mathcal{G}(\hat{E}_{\mathrm{skel}})$, any empirical minimizer satisfies $\hat{G}=G^*$ with
probability tending to one.
\end{proof}

% =============================================================
\section{Sensitivity Analysis to the Threshold}
\label{app:threshold_sensitivity}
% =============================================================

We assess how sensitive S3ME is to the choice of the exceedance threshold, which is parametrized as
$q_{\mathrm{conf}} = 1 - n^{\beta}/n$ for varying $\beta$. We fix $n=1000$, $p=50$, $m=1$, and
repeat 50 Monte Carlo replicates for each $\beta \in \{0.5, 0.6, 0.7, 0.8, 0.9\}$, reporting mean
skeleton F1/SHD and DAG F1/SHD. As $\beta$ decreases, the threshold $q_{\mathrm{conf}}$ rises and
the effective tail sample size $k$ shrinks, making skeleton estimation harder, consistent with
Theorem~3 which requires $k$ to be large relative to $\log p$ for reliable sure screening. The
orientation step exploits an asymptotic directional signal whose existence is guaranteed by
Theorem~1 independently of $k$, and is therefore less sensitive to the threshold choice.
Figure~\ref{fig:threshold_sensitivity} confirms both predictions. Skeleton F1 declines from $0.521$
at $\beta=0.9$ to $0.435$ at $\beta=0.5$, and skeleton SHD rises from $84.3$ to $114.5$
correspondingly, while DAG F1 remains stable at $0.869$ and DAG SHD at $11.1$ across
$\beta \in \{0.6, 0.7, 0.8, 0.9\}$, with only mild degradation at $\beta=0.5$ where $k$ is
smallest. This confirms that the orientation stage is robust to the threshold choice, and that the
default $\beta=0.7$ sits comfortably in the stable region.

\begin{figure}[htbp]
  \centering
  \includegraphics[width=\textwidth]{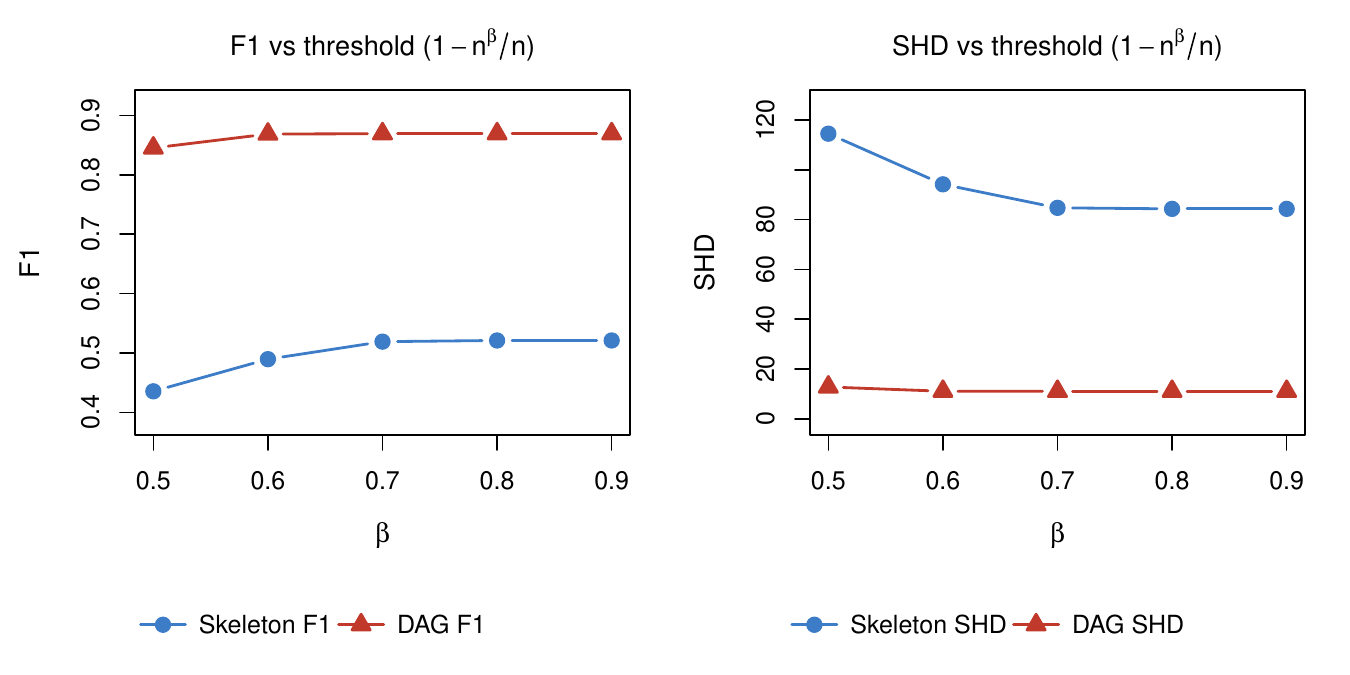}
  \caption{Threshold sensitivity analysis ($n=1000$, $p=50$, 50 replicates).
    \emph{Left:} mean F1 vs.\ $\beta$ for the skeleton step (blue) and the full DAG (red).
    \emph{Right:} mean SHD vs.\ $\beta$ for the skeleton (blue) and the full DAG (red).
    DAG-level metrics are stable across all thresholds; skeleton metrics show mild degradation
    only at large $\beta$ where fewer tail observations are retained.}
  \label{fig:threshold_sensitivity}
\end{figure}

% =============================================================
\section{Robustness to Graph Structure}
\label{app:graph_robustness}
% =============================================================

We examine whether the results from Section~5 depend on the choice of graph generating mechanism.
All parameters are fixed at $n=1000$, $p=50$, $\mathrm{conf\_s}=10$, $C=1.0$,
$\gamma_{\mathrm{EBIC}}=10$, and $50$ Monte Carlo replicates. We compare Barab\'asi--Albert (BA)
graphs with Erd\H{o}s--R\'enyi (ER) graphs \citep{ErdosRenyi1959}, where the ER edge probability is
set to $p_{\mathrm{edge}}=2m/(p-1)$ to align expected degree across graph types. Two attachment
parameters are considered, $m=1$ and $m=2$.

Table~\ref{tab:graph_robustness} reports skeleton and DAG metrics for three additional
configurations, with the $m=1$ BA results from Section~5 serving as the reference. DAG precision
remains near $1.0$ throughout, confirming that the orientation and pruning steps introduce very few
spurious edges regardless of graph structure or attachment parameter. Performance differences are
concentrated in recall. Under ER graphs, the uniform degree distribution allows the lasso to
identify true neighbours more reliably, yielding DAG F1 of $0.979$ and $0.887$ for $m=1$ and $m=2$
respectively. The modest drop from $m=1$ to $m=2$ under ER reflects the expected difficulty of
screening denser graphs. Under BA with $m=2$, however, the same increase in density leads to a much
sharper decline, with DAG F1 falling to $0.644$ and DAG recall to $0.478$. Since ER $m=2$ and BA
$m=2$ have the same expected degree, this gap isolates the effect of the BA hub structure, where
high-degree nodes create local neighbourhoods that are substantially harder to screen, causing true
edges to be missed at the skeleton stage. The deficit then propagates to the DAG level. This confirms
that the performance gap is attributable to graph topology rather than to a limitation of the method
itself.

\begin{table}[htbp]
\centering
\caption{Graph structure robustness ($n=1000$, $p=50$, 50 replicates). Standard deviations in
parentheses. The $m=1$ BA baseline from Section~5 is omitted to avoid repetition.}
\label{tab:graph_robustness}
\begin{tabular}{llcccc}
\toprule
Setting & Stage & Precision & Recall & F1 & SHD \\
\midrule
\multirow{2}{*}{ER, $m=1$} & Skeleton & 0.212 (0.032) & 0.963 (0.031) & 0.346 (0.042) & 182.82 (17.12) \\
                            & DAG      & 1.000 (0.000) & 0.959 (0.032) & 0.979 (0.017) & 2.16 (1.89) \\
\midrule
\multirow{2}{*}{ER, $m=2$} & Skeleton & 0.366 (0.045) & 0.814 (0.062) & 0.502 (0.040) & 159.20 (20.86) \\
                            & DAG      & 1.000 (0.000) & 0.800 (0.063) & 0.887 (0.039) & 20.12 (7.59) \\
\midrule
\multirow{2}{*}{BA, $m=2$} & Skeleton & 0.217 (0.023) & 0.522 (0.055) & 0.306 (0.032) & 229.72 (14.25) \\
                            & DAG      & 0.998 (0.007) & 0.478 (0.056) & 0.644 (0.053) & 50.74 (5.59) \\
\bottomrule
\end{tabular}
\end{table}

% =============================================================
\section{Application Preprocessing and Tuning Details}
\label{app:application_preprocessing}
% =============================================================

\subsection{Danube preprocessing}
For the Danube application, we use daily summer (June--August) discharge measurements from 31
gauging stations over the common period 1960--2010. To reduce short-range temporal dependence, we
apply standard temporal declustering before marginal standardization. Each declustered margin is
then transformed to unit Pareto scale by the empirical rank transform, and the log-tail
coordinates are constructed as in Section~2.3 of the main paper. The application threshold is set
at $q_{\mathrm{conf}}=0.90$, which yields $k=117$ tail observations.

\subsection{S\&P 500 preprocessing and tuning}
For the S\&P 500 application, the raw returns panel is obtained from Stooq over January 2005 to
March 2025. Stocks with more than 5\% missing observations are removed before analysis, leaving
$p=103$ constituents. Loss variables are defined by $L_{t,j}=-r_{t,j}$, after which each marginal
series is transformed to unit Pareto scale by the empirical rank transform and the tail sample is
selected using $q_{\mathrm{conf}} = 1 - n^{0.7}/n$, yielding $k=4{,}120$ observations.

The skeleton is estimated by proxy-adjusted nodewise extremal LASSO with
$\lambda = c\sqrt{\log(p+1)/k}$. Figure~\ref{fig:app2_bic_supp} reports the total BIC score and
resulting DAG sparsity across the tuning grid for the multiplier $c$. The elbow criterion selects
$c=5$, which yields the value used in the main-text application.

\begin{figure}[htbp]
    \centering
    \includegraphics[width=0.95\textwidth]{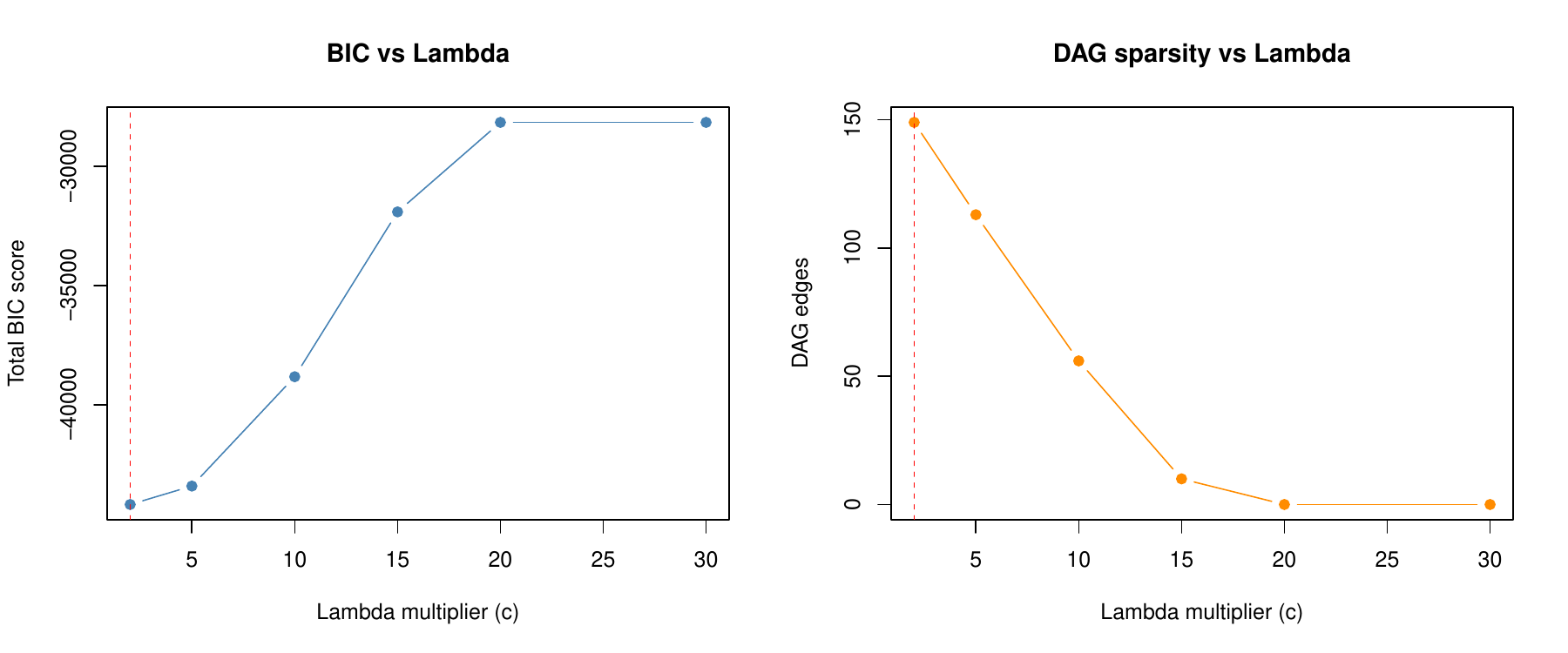}
    \caption{Regularization parameter selection for the S\&P 500 application. Left panel shows the
    total BIC score as a function of the lambda multiplier $c$, and the elbow criterion selects
    $c=5$, corresponding to $\lambda = c\sqrt{\log(p+1)/k}$. Right panel shows the number of DAG
    edges as a function of $c$, illustrating how the estimated graph becomes sparser as the penalty
    increases. The selected $c=5$ is indicated by the red dashed line in both panels.}
    \label{fig:app2_bic_supp}
\end{figure}

\bibliography{reference.bib}

\end{document}